\documentclass[a4paper,reqno]{amsart}

\usepackage[%
pdfauthor={},%
pdftitle={Tunable Localisation in Parity-Time-Symmetric Resonator Arrays with Imaginary Gauge Potentials},
]{hyperref}
\usepackage[utf8]{inputenc}
\usepackage[english]{babel}
\usepackage{standalone}
\usepackage{graphicx}
\graphicspath{ {graphics/} }
\usepackage{imakeidx}
\usepackage{setspace}
\usepackage{appendix}
\usepackage{pdfpages}
\makeindex
\usepackage{mathtools}
\usepackage{tabulary}
\usepackage{booktabs}
\usepackage{amsmath}
\usepackage{amssymb}
\usepackage{amsthm}
\usepackage{amsfonts}
\usepackage{tikz-cd}
\usepackage{extarrows}
\usepackage{tabularx}
\usepackage{float}
\usepackage{enumitem}
\usepackage{quiver}
\usepackage{csquotes}
\restylefloat{table}
\usepackage{bm}

\usepackage{tikz}
\usetikzlibrary{calc,intersections}
\usepackage{tikz-layers}
\usepackage[arrow, matrix, curve]{xy}
\usepackage[final]{microtype}

\makeatletter
\newsavebox\myboxA
\newsavebox\myboxB
\newlength\mylenA

\newcommand*\xoverline[2][0.75]{%
  \sbox{\myboxA}{$\m@th#2$}%
  \setbox\myboxB\null
  \ht\myboxB=\ht\myboxA%
  \dp\myboxB=\dp\myboxA%
  \wd\myboxB=#1\wd\myboxA
  \sbox\myboxB{$\m@th\overline{\copy\myboxB}$}
  \setlength\mylenA{\the\wd\myboxA}
  \addtolength\mylenA{-\the\wd\myboxB}%
  \ifdim\wd\myboxB<\wd\myboxA%
    \rlap{\hskip 0.5\mylenA\usebox\myboxB}{\usebox\myboxA}%
  \else
    \hskip -0.5\mylenA\rlap{\usebox\myboxA}{\hskip 0.5\mylenA\usebox\myboxB}%
  \fi}
\makeatother

\newcommand{\mc}[1]{\mathcal{#1}}

\newcommand{\cm}{C^\gamma}
\newcommand{\cg}{\mathcal{C}^{\theta,\gamma}}
\newcommand{\cu}{C^U_\pm}
\newcommand{\cl}{C^L_\pm}

\newcommand*{\pnr}[1]{\frac{P_N(#1)}{P_{N-1}(#1)}}
\DeclareMathOperator\Arg{Arg}

\usepackage[T1]{fontenc} 
\usepackage{lmodern} 
\rmfamily 
\DeclareFontShape{T1}{lmr}{b}{sc}{<->ssub*cmr/bx/sc}{}
\DeclareFontShape{T1}{lmr}{bx}{sc}{<->ssub*cmr/bx/sc}{}
\usepackage{orcidlink}
\usepackage{lipsum}
\usepackage{mathtools}
\usepackage{amsmath}
\usepackage{amssymb}
\usepackage{mathdots}
\usepackage{tikz}
\usetikzlibrary{matrix,positioning,decorations.pathreplacing,calc}
\usetikzlibrary{
    decorations.text,%
    decorations.markings,%
    shadows}
\usepackage{adjustbox}
\usepackage{graphicx}

\usepackage{caption}
\usepackage{subcaption}

\usepackage{dsfont} 
\usepackage[format=hang]{caption}

\usepackage[normalem]{ulem}

\newcommand{\abs}[1]{\left\lvert#1\right\rvert}
\usepackage{bm}

\usepackage[
    style=numeric,
    sorting=nty,
    maxnames=99,
    maxalphanames=5,
    natbib=true,
    sortcites]{biblatex}

\DeclareNameAlias{default}{family-given}

\graphicspath{{figures/}}

\AtEveryBibitem{
    \clearfield{url}
    \clearfield{issn}
    \clearfield{isbn}
    \clearfield{urldate}

    \ifentrytype{book}{
        \clearfield{pages}}{%
    }
}

\usepackage{xargs}                      
\usepackage[prependcaption,textsize=tiny,textwidth=3cm]{todonotes}
\newcommandx{\unsure}[2][1=]{\todo[linecolor=red,backgroundcolor=red!25,bordercolor=red,#1]{#2}}
\newcommandx{\change}[2][1=]{\todo[linecolor=blue,backgroundcolor=blue!25,bordercolor=blue,#1]{#2}}
\newcommandx{\info}[2][1=]{\todo[linecolor=OliveGreen,backgroundcolor=OliveGreen!25,bordercolor=OliveGreen,#1]{#2}}
\newcommandx{\improvement}[2][1=]{\todo[linecolor=black,backgroundcolor=black!25,bordercolor=black,#1]{#2}}
\newcommandx{\thiswillnotshow}[2][1=]{\todo[disable,#1]{#2}}
\usepackage{amsthm}
\usepackage[noabbrev,capitalize]{cleveref}
\crefname{proposition}{Proposition}{Propositions}
\crefname{equation}{}{}

\newtheorem{theorem}{Theorem}[section]
\newtheorem{lemma}[theorem]{Lemma}
\newtheorem{proposition}[theorem]{Proposition}
\newtheorem{corollary}[theorem]{Corollary}

\theoremstyle{definition}
\newtheorem{definition}[theorem]{Definition}

\newtheorem{remark}[theorem]{Remark}

\crefname{assumption}{Assumption}{Assumptions}
\crefname{definition}{Definition}{Definitions}
\crefname{corollary}{Corollary}{Corollaries}
\crefname{enumi}{item}{items}

\usepackage{fancyhdr}

\fancypagestyle{plain}{%
    \fancyhead[C]{}
    \fancyfoot[C]{}
    
    }
\fancypagestyle{nosection}{%
    \fancyhead[CE]{}
    \fancyhead[CO]{}
    \fancyfoot[CE,CO]{\thepage}
}

\DeclareMathOperator{\N}{\mathbb{N}}

\DeclareMathOperator{\R}{\mathbb{R}}
\DeclareMathOperator{\C}{\mathbb{C}}
\DeclareMathOperator{\sgn}{sign}

\renewcommand{\i}{\mathbf{i}}

\newcommand{\inv}{^{-1}}

\renewcommand{\qedsymbol}{$\blacksquare$}

\usepackage{fancyhdr}

\fancypagestyle{plain}{%
    \fancyhead[C]{}
    \fancyfoot[C]{}
    
    }
\fancypagestyle{nosection}{%
    \fancyhead[CE]{}
    \fancyhead[CO]{}
    \fancyfoot[CE,CO]{\thepage}
}

\DeclareMathOperator{\diag}{diag}
\DeclareMathOperator{\BO}{\mathcal{O}}

\DeclareMathOperator{\capmat}{\mathcal{C}}

\DeclareMathOperator{\capmatg}{\mathcal{C}^\gamma}

\renewcommand{\epsilon}{\varepsilon}
\DeclareMathOperator{\dd}{d\!}

\let\emptyset\varnothing

\renewcommand{\i}{\mathbf{i}}

\DeclareMathOperator{\iL}{{\mathsf{L}}}
\DeclareMathOperator{\iR}{{\mathsf{R}}}
\DeclareMathOperator{\iLR}{{\mathsf{L},\mathsf{R}}}

\usepackage{tcolorbox}
\usetikzlibrary{tikzmark,calc}

\usepackage{upgreek}
\numberwithin{equation}{section}

\title[Tunable Localisation in Non-Hermitian Resonator Arrays]{Tunable Localisation in Parity-Time-Symmetric Resonator Arrays with Imaginary Gauge Potentials}

\begin{document}
 \author[H. Ammari]{Habib Ammari\,\orcidlink{0000-0001-7278-4877}}
 \address{\parbox{\linewidth}{Habib Ammari\\
  ETH Z\"urich, Department of Mathematics, Rämistrasse 101, 8092 Z\"urich, Switzerland, \href{http://orcid.org/0000-0001-7278-4877}{orcid.org/0000-0001-7278-4877}}.}
 \email{habib.ammari@math.ethz.ch}
 \thanks{}

 \author[S. Barandun]{Silvio Barandun\,\orcidlink{0000-0003-1499-4352}}
  \address{\parbox{\linewidth}{Silvio Barandun\\
  ETH Z\"urich, Department of Mathematics, Rämistrasse 101, 8092 Z\"urich, Switzerland, \href{http://orcid.org/0000-0003-1499-4352}{orcid.org/0000-0003-1499-4352}}.}
  \email{silvio.barandun@sam.math.ethz.ch}

\author[P. Liu]{Ping Liu\,\orcidlink{0000-0002-7857-7040}}
 \address{\parbox{\linewidth}{Ping Liu\\
 ETH Z\"urich, Department of Mathematics, Rämistrasse 101, 8092 Z\"urich, Switzerland, \href{http://orcid.org/0000-0002-7857-7040}{orcid.org/0000-0002-7857-7040}}.}
\email{ping.liu@sam.math.ethz.ch}

 \author[A. Uhlmann]{Alexander Uhlmann\,\orcidlink{0009-0002-0426-6407}}
  \address{\parbox{\linewidth}{Alexander Uhlmann\\
  ETH Z\"urich, Department of Mathematics, Rämistrasse 101, 8092 Z\"urich, Switzerland, \href{http://orcid.org/0009-0002-0426-6407}{orcid.org/0009-0002-0426-6407}}.}
  \email{alexander.uhlmann@sam.math.ethz.ch}
  
\begin{abstract}
The aim of this paper is to illustrate both analytically and numerically the interplay of two fundamentally distinct non-Hermitian mechanisms in a deep subwavelength regime. 
Considering a parity-time symmetric system of one-dimensional subwavelength resonators equipped with two kinds of non-Hermiticity --- an imaginary gauge potential and on-site gain and loss --- we prove that all but two eigenmodes of the system decouple when going through an exceptional point. By tuning the gain-to-loss ratio, the system changes from a phase with unbroken parity-time symmetry to a phase with broken parity-time symmetry. At the macroscopic level,  this is observed as a transition from symmetrical eigenmodes to condensated eigenmodes at one edge of the structure. Mathematically, it arises from a topological state change. The results of this paper open the door to the justification of a variety of phenomena arising from the interplay between non-Hermitian reciprocal and non-reciprocal mechanisms not only in subwavelength wave physics but also in quantum mechanics where the tight binding model coupled with the nearest neighbour approximation can be analysed with the same tools as those developed here.
\end{abstract}

\maketitle

\date{}

\bigskip

\par

\bigskip

\\

\section{Introduction}
 In this paper, we study the interplay of two fundamentally distinct non-Hermitian wave mechanisms in a deep subwavelength regime using first-principles mathematical analysis.  The ultimate goal of subwavelength wave physics is to manipulate waves at subwavelength scales. Recent breakthroughs, such as the emergence of the field of metamaterials, have allowed us to do this in a way that is robust, possibly non-reciprocal,  and that beats diffraction limits. Spectacular properties of metamaterials such as super-focusing, super-resolution, waves with exponentially growing amplitudes, Anderson-type localisation at deep subwavelength scales, unidirectional invisibility and cloaking, single and double near-zero effective properties have been recently rigorously explained; see, for instance,  \cite{ammari.davies.ea2022Exceptional,ammari.barandun.ea2023Edge,ammari.barandun.ea2024Mathematical,ouranderson2022,ammari.davies.ea2021Functional,ammari.cao.ea2023Transmission,qiu.lin.ea2023Mathematical,thiang.zhang2023Bulkinterface,li.lin.ea2023Dirac,lin.zhang2022Mathematicala}. A variety of Hermitian, non-Hermitian, and time-modulated systems of subwavelength resonators have been considered. Phase transitions and degeneracies in the mathematical structures of those models which are responsible for exotic phenomena have been identified. Here, we consider one-dimensional systems of high-contrast subwavelength resonators as a demonstrative setting to develop a mathematical and numerical framework for the interplay of reciprocal gain-loss and non-reciprocal mechanisms in the subwavelength regime.
 
The concept of non-Hermitian physics, originally developed in the context of quantum field theory \cite{bender.boettcher1998Real}, has been investigated on distinct classical wave platforms and created a plethora of counter-intuitive phenomena \cite{el-ganainy.makris.ea2018NonHermitiana,ashida.gong.ea2020NonHermitian}.  In subwavelength wave physics, non-Hermiticity can be obtained via either a reciprocal mechanism by adding gain and loss inside the resonators \cite{miri.alu2019Exceptional}, or via a non-reciprocal one by introducing a directional damping term, which is motivated by an imaginary gauge potential \cite{yokomizo.yoda.ea2022NonHermitian}. 

On the one hand, introducing gain and loss inside the resonators, represented by the imaginary parts of complex-valued material parameters, can create exceptional points. An exceptional point is a point in parameter space at which two or more eigenstates coalesce  \cite{heiss2012physics,miri.alu2019Exceptional,ammari.davies.ea2022Exceptional,haiexceptional}. A degeneracy of this nature gives rise to structures with remarkable properties such as high sensitivity \cite{ammari2020high,hodaei2017enhanced, vollmer2008single}. As is common in the field of non-Hermitian physics, we will consider structures with \emph{parity--time} ($\mc{PT}$-) \emph{symmetry}, which forces the spectrum of the governing operator to be conjugate-symmetric. Exceptional points are then the transition points between a real spectrum and a non-real spectrum which is symmetric around the real axis. They are a consequence of balanced symmetries in the system, which cause the eigenvectors to align. So that the system already has some underlying symmetry, exceptional points are often sought in structures with \emph{parity--time symmetry}. 
	
On the other hand, for systems that are non-Hermitian due to non-reciprocity, the wave propagation is usually amplified in one direction and attenuated in the other. This inherent unidirectional dynamics is
related to the non-Hermitian skin effect, which leads to the accumulation of modes at one edge of the structure \cite{hatano,yokomizo.yoda.ea2022NonHermitian,rivero.feng.ea2022Imaginary}. Recently, it was proved in \cite{ammari.barandun.ea2024Mathematical} that the spectrum is real and the exponential decay of eigenmodes and their accumulation at one edge of the structure are induced by the Fredholm index of an associated Toeplitz operator. Moreover, it was shown in \cite{jana.sirota2023Emerging} that a tunnelling-like phenomenon occurs when connecting two non-Hermitian chains with mirrored non-reciprocity. 

The aim of this paper is to consider a mirrored system with two imaginary gauge potentials (opposite to each other) and study the phase change of the spectrum from purely real to complex when gain and loss are introduced in a balanced way into the system as a function of the gain to loss ratio. This ensures that parity--time symmetry is preserved as this ratio is increased. Using asymptotic methodology that was developed in \cite{ammari.barandun.ea2024Mathematical},  we can approximate the subwavelength resonant modes by the eigenvalues of a so-called \emph{gauge capacitance matrix}.
Crucially, the parity--time symmetry of the system is reflected in the gauge capacitance matrix $C$, ensuring it is \emph{pseudo--Hermitian}, that is there exists some invertible self-adjoint matrix $M$ so that the adjoint $C^*$ of $C$ is given by $C^*=MCM^{-1}$.
Our main contribution in this work is to prove that the eigenmodes of the parity--time symmetric system decouple when going through an exceptional point. Tuning the gain-to-loss ratio, we change the system from a phase with unbroken parity-time symmetry to a phase with broken parity-time symmetry where the condensed eigenmodes at one edge are decoupled from the ones at the opposite edge of the structure. To understand this behaviour we extend the standard Toeplitz theory to encompass symmetrical parameter changes across an interface. We show that the intrinsic nature of this switch from unbroken to broken $\mc{PT}$-symmetry is due to a change in the topological nature of the mode.
Furthermore, we are able to show that as the number of resonators is increased, the amount of tuning required for exceptional points and the corresponding decoupling to occur goes to zero. This leads to an increasingly dense concentration of exceptional points. As the tight-binding model in quantum mechanics when coupled 
with the nearest neighbour approximation reduces to the study of a tridiagonal Toeplitz matrix, the tools developed in this paper lead to similar results for non-Hermitian quantum systems. 


The paper is structured as follows. In \cref{sec:setup}, we introduce the physical setup in all due details,  recall what is known in the literature for similar systems,  and fix the notation. Crucially, eigenfrequencies and eigenmodes can be approximated by eigenpairs of a finite-dimensional linear operator. \cref{sec:chebypoly} gives a first characterisation of the eigenpairs in terms of Chebyshev polynomials. This gives sufficient and necessary conditions for the eigenvalues and eigenvectors of the aforementioned matrix; nevertheless the conditions are not explicit. In \cref{sec:eva}, we study the eigenvalues more closely,  characterise them precisely, and show the existence of exceptional points. Building on understanding of the eigenvalues, \cref{sec:eves} analyses the eigenvectors and proves a complete characterisation of their macroscopic nature (exponential decay/growth) based on the Toeplitz index of a related operator.
In \cref{sec:matrixsymmetries}, we recall some well-known results on matrix symmetries. In \cref{sec:0gain}, we reduce the problem of finding the eigenfrequencies in the case where there is no gain or loss introduced into the system to finding the spectrum of a tridiagonal almost Toeplitz matrix. \cref{sec:poly_inter} and \cref{sec:technical_proofs} are dedicated to the proofs of some technical results.

\section{Setup} \label{sec:setup}
Here, we assume the same setting as in \cite{ammari.barandun.ea2024Mathematical,ammari.barandun.ea2023Edge}. We consider a one-dimensional chain of $N$ disjoint identical subwavelength resonators $D_i\coloneqq (x_i^{\iL},x_i^{\iR})$, where $(x_i^{\iLR})_{1\<i\<N} \subset \R$ are the $2N$ extremities satisfying $x_i^{\iL} < x_i^{\iR} <  x_{i+1}^{\iL}$ for any $1\leq i \leq N$. We fix the coordinates such that $x_1^{\iL}=0$. We also denote by  $\ell_i = x_i^{\iR} - x_i^{\iL}$ the length of each of the resonators,  and by $s_i= x_{i+1}^{\iL} -x_i^{\iR}$ the spacing between the $i$-th and $(i+1)$-th resonators. We use
\begin{align*}
    D\coloneqq \bigcup_{i=1}^N(x_i^{\iL},x_i^{\iR})
\end{align*}
to symbolise the set of subwavelength resonators. In this paper, we only consider systems of equally spaced identical resonators, that is,
\begin{align*}
    \ell_i = \ell \in \R_{>0}\text{ for all } 1\leq i\leq N \quad \text{and} \quad s_i = s \in \R_{>0}  \text{ for all } 1\leq i\leq N-1.
\end{align*}
This will simplify the formulas in our subsequent analysis and is sufficient to understand the fundamental mechanisms behind the non-Hermitian effects we are interested in.

In this work, we consider the following one-dimensional
damped wave equation where the damping acts in the space dimension instead of the time dimension:
\begin{align}
    -\frac{\omega^{2}}{\kappa(x)}u(x)- \gamma(x) \frac{\dd}{\dd x}u(x)-\frac{\dd}{\dd x}\left( \frac{1}{\rho(x)}\frac{\dd}{\dd
        x}  u(x)\right) =0,\qquad x \in\R,
    \label{eq: gen Strum-Liouville}
\end{align}
for a piecewise constant damping coefficient
\begin{align}\label{equ:nonhermitiancoeff1}
    \gamma(x) = \begin{dcases}
        \gamma_i,\quad x\in D_i, \\
        0,\quad x \in \R\setminus D.
    \end{dcases}
\end{align}
The parameters $\gamma_i$ extend the usual scalar wave equation to a generalised Strum--Liouville equation via the introduction of an imaginary gauge potential \cite{yokomizo.yoda.ea2022NonHermitian}.
The material parameters $\kappa(x)$ and $\rho(x)$ are piecewise constant
\begin{align*}
    \kappa(x)=
    \begin{dcases}
        \kappa_i, & x\in D_i,          \\
        \kappa,   & x\in\R\setminus D,
    \end{dcases}\quad\text{and}\quad
    \rho(x)=
    \begin{dcases}
        \rho_b, & x\in D,            \\
        \rho,   & x\in\R\setminus D,
    \end{dcases}
\end{align*}
where the constants $\rho_b, \rho, \kappa, \in \R_{>0}$ and $\kappa_i \in \C$. The wave speeds inside the resonators $D$ and inside the background medium $\R\setminus D$, are denoted respectively by $v_i$ and $v$,
the wave numbers respectively by $k_b$ and $k$, the frequency by $\omega$, and the contrast between the densities of the resonators and the background medium by $\delta$:
\begin{align}
    v_i:=\sqrt{\frac{\kappa_i}{\rho_b}}, \qquad v:=\sqrt{\frac{\kappa}{\rho}},\qquad
    k_i:=\frac{\omega}{v_i},\qquad k:=\frac{\omega}{v},\qquad
    \delta:=\frac{\rho_b}{\rho}.
\end{align}

We are interested in the resonances $\omega\in\C$ such that \eqref{eq: gen Strum-Liouville} has a non-trivial solution in a high-contrast, low-frequency (subwavelength) regime. This regime is typically characterised by letting the contrast parameter $\delta\to 0$ and looking for solutions which are such that $\omega \to 0$ as $\delta\to 0$. One consequence of this asymptotic ansatz is that it lends itself to characterisation using asymptotic analysis \cite{ammari.davies.ea2021Functional}. Note that this limit recovers subwavelength resonances, while keeping the size of the resonators fixed.

In \cite{ammari.barandun.ea2024Mathematical}, an asymptotic analysis in the subwavelength limit was performed on the system of non-Hermitian, non-reciprocal, one-dimensional subwavelength resonators. The setup considered there was simpler: all resonators had the same imaginary gauge potential and real material parameters. For such a system, it was shown that the resonances are given by the eigenstates of the \emph{gauge capacitance matrix} $\capmat^\gamma$. This is a modified version of the conventional capacitance matrix that is often used to characterise many-body low-frequency resonance problems; see, for instance, \cite{ammari.davies.ea2021Functional}. \cref{thm:gauge} summarises the main result of \cite{ammari.barandun.ea2024Mathematical}.
\begin{theorem} \label{thm:gauge}
    Consider a system of $N$ identical and equally spaced resonators all with the same imaginary gauge potential $\gamma$ and the wave speed $v_i=v_b= \sqrt{\kappa_b/\rho_b}$ for all $1\leq i\leq N$. Let the gauge capacitance matrix $\capmat^\gamma =(\capmat_{i,j}^\gamma)^N_{i,j=1}$ be defined by
    \begin{align}
        \label{eq: cap mat ESI}
        \capmat_{i,j}^\gamma \coloneqq \begin{dcases}
            \frac{\ell\gamma}{s} \frac{1}{1-e^{-\gamma\ell}},        & i=j=1,                 \\
            \frac{\ell\gamma}{s} \coth(\gamma\ell/2),                & 1< i=j< N,             \\
            \pm\frac{\ell\gamma}{ s} \frac{1}{1-e^{\pm\gamma \ell}}, & 1\leq i=j \pm 1\leq N, \\
            -\frac{\ell\gamma}{ s} \frac{1}{1-e^{\gamma\ell}},       & i=j=N,                 \\
            \ 0,                                                 & \text{else}.
        \end{dcases}
    \end{align}
    Then,
    \begin{itemize}
    \item The $N$ subwavelength eigenfrequencies $\omega_i$ of (\ref{eq: gen Strum-Liouville}) associated to this system satisfy, as $\delta\to0$,
              \begin{align*}
                  \omega_i =  \sqrt{\delta\lambda_i} + \BO(\delta),
              \end{align*}
              where $(\lambda_i)_{1\leq i\leq N}$ are the eigenvalues of the eigenvalue problem $$
              VL^{-1}\capmat^\gamma \bm a_i = \lambda_i \bm a_i
              $$
              with $V =v_b^2I_{N}$ and $L=\ell I_{N}$.
              Furthermore, let $u_i$ be a subwavelength eigenmode corresponding to $\omega_i$ and let $\bm a_i$ be the corresponding eigenvector of $VL^{-1}\capmatg$. Then, 
              \begin{align*}
                  u_i(x) = \sum_j \bm a_i^{(j)}V_j(x) + \BO(\delta),
              \end{align*}
              where $V_j(x)$ are defined by
              \begin{align}
                  \begin{dcases}
                      -\frac{\dd{^2}}{\dd x^2} V_i(x) =0, & x\in\R\setminus\bigcup_{i=1}^N(x_i^{\iL},x_i^{\iR}), \\
                      V_i(x)=\delta_{ij},              & x\in (x_j^{\iL},x_j^{\iR}),                          \\
                      V_i(x) = \BO(1),                  & \mbox{as } \abs{x}\to\infty.
                  \end{dcases}
                  \label{eq: def V_i}
              \end{align}
        \item All the eigenvalues of $VL^{-1}\capmat^\gamma$ are real. They are given by
              \begin{align}
                  \lambda_1 & = 0,\nonumber                                                                                                                                                                                                   \\
                  \lambda_k & = \frac{\gamma}{s} \coth(\gamma\ell/2)+\frac{2\abs{\gamma}}{s}\frac{e^{\frac{\gamma\ell}{2}}}{\vert e^{\gamma\ell}-1\vert}\cos\left(\frac{\pi}{N}k\right), \quad 2\leq k\leq N . \label{eq: eigenvalues capmat}
              \end{align}
              Furthermore, for $2\leq k\leq N$, the associated eigenvectors $\bm a_k$ satisfy the following inequality:
              \begin{align}
                  \vert \bm a_k^{(i)}\vert \leq \kappa_k e^{-\gamma\ell\frac{i-1}{2}}\quad \text{for all } 1\leq i\leq N \label{eq: decay for eigemodes},
              \end{align}
              for some $\kappa_k\leq (1+e^{\frac{\gamma\ell}{2}})^2$. Here,
              $\bm a_k^{(i)}$ denotes the $i$-th entry of the eigenvector  $\bm a_k$.
    \end{itemize}
\end{theorem}
In this paper, we consider the following setup:
\begin{align*}
    v_i=
    \begin{dcases}
        e^{\frac{1}{2}\i\theta},  & 1\leq i\leq N,    \\
        e^{-\frac{1}{2}\i\theta}, & N+1\leq i\leq 2N,
    \end{dcases}\quad\text{and}\quad
    \gamma_i=
    \begin{dcases}
        \gamma,  & 1\leq i\leq N,    \\
        -\gamma, & N+1\leq i\leq 2N,
    \end{dcases}
\end{align*}
for some fixed $\gamma>0$ and $0\leq \theta \leq 2\pi$. By the periodicity of $e^{\frac{1}{2}\i\theta}$, it will be sufficient to focus on the range $0\leq \theta \leq \frac{\pi}{2}$. Because the resonator length $\ell$ can be absorbed into $\gamma$ we may also assume $s=\ell=1$ without loss of generality.
The system is illustrated in \cref{fig:setting}.
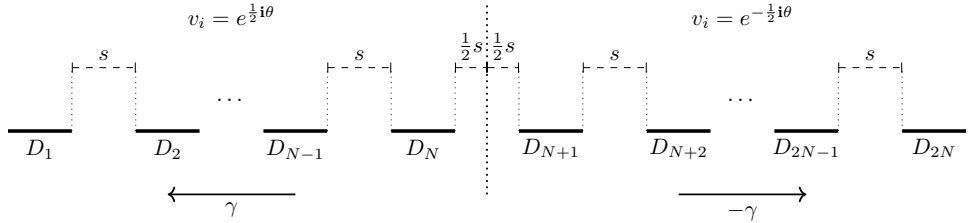
\begin{figure}[htb]
    \centering
    \begin{adjustbox}{width=\textwidth}
        \begin{tikzpicture}

            \begin{scope}[shift={(-7.5,0)}]
                \draw[ultra thick] (0,0) -- node[below]{$D_{1}$} (1,0);

                \draw[-,dotted] (1,0) -- (1,1);
                \draw[|-|,dashed] (1,1) -- node[above]{$s$} (2,1);
                \draw[-,dotted] (2,0) -- (2,1);
            \end{scope}

            \begin{scope}[shift={(-5.5,0)}]
                \draw[ultra thick] (0,0) -- node[below] {$D_{2}$} (1,0);
                \node at (1.5,.5) {\dots};
            \end{scope}

            \begin{scope}[shift={(-3.5,0)}]
                \draw[ultra thick] (0,0) -- node[below]{$D_{N-1}$} (1,0);

                \draw[-,dotted] (1,0) -- (1,1);
                \draw[|-|,dashed] (1,1) -- node[above]{$s$} (2,1);
                \draw[-,dotted] (2,0) -- (2,1);
            \end{scope}

            \begin{scope}[shift={(-1.5,0)}]
                \draw[ultra thick] (0,0) -- node[below] {$D_{N}$} (1,0);
                \draw[-,dotted] (1,0) -- (1,1);
                \draw[|-|,dashed] (1,1) -- node[above]{$\frac{1}{2}s$} (1.5,1);
            \end{scope}

            \draw[->,thick] (-3,-1) -- node[below]{$\gamma$} (-5,-1);
            \node[above] at (-4,1.5) {$v_i = e^{\frac{1}{2}\i\theta}$};

            \draw[-,thick,dotted] (0,-1) -- (0,2);

            \begin{scope}[shift={(-0.5,0)}]
                \draw[|-|,dashed] (0.5,1) -- node[above]{$\frac{1}{2}s$} (1,1);
                \draw[ultra thick] (1,0) -- (2,0);
                \node[below] at (1.5,0) {$D_{N+1}$};
                \draw[-,dotted] (1,0) -- (1,1);
            \end{scope}

            \begin{scope}[shift={(+1.5,0)}]
                \draw[-,dotted] (0,0) -- (0,1);
                \draw[|-|,dashed] (0,1) -- (1,1);
                \node[above] at (0.5,1) {$s$};
                \draw[ultra thick] (1,0) -- (2,0);
                \node[below] at (1.5,0) {$D_{N+2}$};
                \draw[-,dotted] (1,0) -- (1,1);
                \node at (2.5,.5) {\dots};
            \end{scope}

            \begin{scope}[shift={(+3.5,0)}]
                \draw[ultra thick] (1,0) -- (2,0);
                \node[below] at (1.5,0) {$D_{2N-1}$};
            \end{scope}

            \begin{scope}[shift={(+5.5,0)}]
                \draw[-,dotted] (0,0) -- (0,1);
                \draw[|-|,dashed] (0,1) -- (1,1);
                \node[above] at (0.5,1) {$s$};
                \draw[ultra thick] (1,0) -- (2,0);
                \node[below] at (1.5,0) {$D_{2N}$};
                \draw[-,dotted] (1,0) -- (1,1);
            \end{scope}

            \draw[->,thick] (3,-1) -- node[below]{$-\gamma$} (5,-1);
            \node[above] at (4,1.5) {$v_i = e^{-\frac{1}{2}\i\theta}$};
        \end{tikzpicture}
    \end{adjustbox}
    \caption{A chain of $2N$ one-dimensional identical and equally spaced resonators. Material parameters and sign of the imaginary gauge potentials depend on the resonator's position.}
    \label{fig:setting}
\end{figure}
The matrix associated to this structure is slightly different form the one defined in \cref{eq: cap mat ESI}. On one side, the gauge capacitance matrix is given by
\begin{gather}
    \label{eq:cdef}
    C^\gamma =
    \left(\begin{array}{ccccc|ccccc}
            \alpha + \beta & \eta   &        &        &        &        &        &        &        &              \\
            \beta          & \alpha & \ddots &        &        &        &        &        &        &              \\
                           & \ddots & \ddots &        &        &        &        &        &        &              \\
                           &        &        & \alpha & \eta   &        &        &        &        &              \\
                           &        &        & \beta  & \alpha & \eta   &        &        &        &              \\
            \hline
                           &        &        &        & \eta   & \alpha & \beta                                   \\
                           &        &        &        &        & \eta   & \alpha & \ddots                         \\
                           &        &        &        &        &        & \ddots & \ddots                         \\
                           &        &        &        &        &        &        &        & \alpha & \beta        \\
                           &        &        &        &        &        &        &        & \eta   & \alpha+\beta
        \end{array}\right) \in \R^{2N\times 2N}
\end{gather}
with 
\begin{equation}\label{equ:alphabetagamma}
\alpha = \frac{\gamma}{1-e^{-\gamma}} - \frac{\gamma}{1-e^{\gamma}} = \gamma\coth(\gamma/2),\quad \eta = \frac{-\gamma}{1-e^{-\gamma}},\quad \beta = \frac{\gamma}{1-e^{\gamma}},
\end{equation}
because of the sign change of the imaginary gauge potential. On the other side, we have to model the complex (and varying) material parameters. Thus, we consider the \emph{generalised gauge capacitance matrix}
\begin{align}
    \label{eq:cgdef}
    \mc{C}^{\theta,\gamma} = V^\theta C^\gamma \quad \text{ with } \quad V^\theta =
    \left(\begin{array}{c|c}
            e^{\i\theta}I_{N} & \mathbf{0}        \\
            \hline
            \mathbf{0}       & e^{-\i\theta}I_{N}
        \end{array}\right) \in \R^{2N\times 2N}.
\end{align}
The same result as the one stated in the first point of \cref{thm:gauge} holds for the system described by \cref{fig:setting} when considering the generalised gauge capacitance matrix from \eqref{eq:cgdef} (generalising the proof presented in \cite{ammari.barandun.ea2024Mathematical} is easily achieved by the same procedure used in \cite{ammari.barandun.ea2023Edge}). Throughout the paper, $\mc{C}^{\theta,\gamma}$ and $\mc{C}^{\gamma}$ are $2N\times 2N$ matrices. 

This paper will extensively study \emph{non-diagonalisability} of $\cg$.
\begin{definition}
    A setup for which $\cg$ is \emph{not} diagonalisable is called an \emph{exceptional point}.
\end{definition}
Specifically, we will study setups where the geometry and the imaginary gauge potentials remain fixed and the material parameters (here modelled by $\theta$) lay in a specific range.

\subsection{Properties of the gauge capacitance matrix}
We will conclude this section by giving some basic but important properties of the gauge capacitance matrix. Let $P\in \R^{2N\times 2N}$ be the anti-diagonal involution, \emph{i.e.},  $P_{ij} = \delta_{i,2N-i+1}$ and $D^\gamma=\diag(1,e^\gamma,\dots,(e^{\gamma})^{N-1},(e^{\gamma})^{N-1},\dots,e^\gamma,1)\in \R^{2N\times 2N}$. We refer to \cref{sec:matrixsymmetries} for the precise definitions of pseudo-Hermitian and quasi-Hermitian matrices. Here and elsewhere in the paper, $M^*$ denotes the adjoint of $M$: $(M^*)_{i,j}=\overline{M_{j,i}}$.

\begin{proposition}\label{prop:csymm}
    The generalised gauge capacitance matrix has the following symmetry:
    \begin{equation}\label{eq:ptsymm}
        P\cg P = \overline{\cg}.
    \end{equation}
    For the unmodified gauge capacitance matrix $C^\gamma$, we have
    \begin{equation}\label{eq:csymm}
        PC^\gamma P = C^\gamma.
    \end{equation}
\end{proposition}
\begin{proof}
    Equation (\ref{eq:csymm}) follows from noticing that $P$ is symmetric and conjugation by $P$ corresponds to flipping $\cm$ along its diagonal and anti-diagonal.
    Equation (\ref{eq:ptsymm}) follows from equation (\ref{eq:csymm}), $PV^\theta = V^{-\theta} P = \overline{V^\theta} P$, and the fact that $P$ and $\cm$ have real entries.
\end{proof}
\begin{proposition}
    Let $M^{\theta,\gamma} = PV^\theta D^\gamma$.
    Then, $M^{\theta,\gamma}$ is invertible and Hermitian and we have,
    \begin{equation}
        M^{\theta,\gamma} \cg = (\cg)^*M^{\theta,\gamma}.
    \end{equation}
\end{proposition}
In line with \cref{def:pseudoherm}, $\cg$ as in \eqref{eq:cgdef} is pseudo-Hermitian, and its spectrum must be invariant under complex conjugation, that is, 
    \[\sigma(\cg) = \overline{\sigma(\cg)}.\]

For the case $\theta=0$, the matrix $\mc{C}^{0,\gamma} = \cm$ satisfies an even stronger notion of Hermiticity.
\begin{proposition}\label{prop:cquasiherm}
Let $\cm=\mathcal{C}^{\theta=0,\gamma}$ as in equation (\ref{eq:cdef}). Then, $\cm$ is quasi-Hermitian with \emph{metric operator} $D$, that is, 
    \begin{equation}
        D^{-1}\cm  = (\cm)^*D^{-1}.
    \end{equation}
\end{proposition}
From \cref{cor:quasiherm}, we can then immediately see that $\cm$ is diagonalisable with real spectrum.

Finally, we characterise the kernel of $\cg$.
\begin{lemma}
    For any $\gamma>0$ and $\theta \in [0,2\pi)$, we have $(1,\dots,1) \in \ker \cg\subset \R^{2N}$.
\end{lemma}
\begin{proof}
    This follows immediately from the fact that $\alpha+\beta+\eta=0$.
\end{proof}
As we will see in the proof of \cref{thm:spectrarepresentation1}, the eigenspaces of $\cg$ are always one-dimensional. Consequently, the kernel of $\cg$ is also one-dimensional and is exactly the span of $(1,\dots,1)$.

\section{Characterisation of eigenpairs of the generalised gauge capacitance matrix}\label{sec:chebypoly}
In this section, we will exploit the partially Toeplitz structure of $\cg$ to find the general form of its eigenvectors and a characterisation of its eigenvalues in terms of Chebyshev polynomials. The fact that $\cg$ is tridiagonal allows us to determine the eigenvectors recursively. The symmetry of $\cg$ across the interface in the middle will yield very similar forms for the first and second half of its eigenvectors. The eigenvalues will then be characterised by a compatibility condition across this interface.

\begin{theorem}\label{thm:spectrarepresentation1}
    Let the affine transformation $\mu^\theta:\C\to\C$ be defined by
    \begin{align}\label{eq:mudef}
        \mu^\theta(\lambda)\coloneqq \frac{e^{-\i\theta}\lambda-\alpha}{2\sqrt{\beta\eta}}=e^{-\i\theta}\lambda\frac{1}{\gamma}\sinh \frac{\gamma}{2} - \cosh\frac{\gamma}{2}.
    \end{align}
    For $\lambda\in \C$ an eigenvalue of $\cg$, the corresponding eigenvector is given by $\bm v = (\bm x, \bm y)^{\top}$ where
    \begin{equation}\label{eq:evecform}
        \begin{aligned}
             & \bm x=\left(P_0(\mu^\theta(\lambda)), \left(e^{-\frac{\gamma}{2}}\right)P_1(\mu^\theta(\lambda)),  \cdots,  \left(e^{-\frac{\gamma}{2}}\right)^{N-1} P_{N-1}(\mu^\theta(\lambda)) \right)^{\top},            \\
             & \bm y = C\left(\left(e^{-\frac{\gamma}{2}}\right)^{N-1} P_{N-1}(\mu^{-\theta}(\lambda)), \cdots,  \left(e^{-\frac{\gamma}{2}}\right)P_1(\mu^{-\theta}(\lambda)), P_0(\mu^{-\theta}(\lambda)) \right)^{\top}.
        \end{aligned}
    \end{equation}
    Here, $P_n (x) = U_n(x) + e^{-\frac{\gamma}{2}}U_{n-1}(x)$ is the sum of two Chebyshev polynomials of the second kind, with $P_0 = 1$. Specifically, $U_{n+1}(x)\coloneqq 2xU_n(x)-U_{n-1}(x)$ for $n\geq 1 $ with $U_0(x)=1$ and $U_1(x)=2x$.
    Furthermore, we have
    \begin{equation}\label{eq:cdef2}
        C= e^{-\frac{\gamma}{2}}\frac{P_N(\mu^{\theta}(\lambda))}{P_{N-1}(\mu^{-\theta}(\lambda))}= e^{\frac{\gamma}{2}}\frac{P_{N-1}(\mu^{\theta}(\lambda))}{P_N(\mu^{-\theta}(\lambda))},
    \end{equation}
    which yields the following characterisation of the spectrum of $\cg$:
    \begin{equation}\label{eq:evalconstraint}
        \frac{P_{N}(\mu^{\theta}(\lambda))P_{N}(\mu^{-\theta}(\lambda))}{P_{N-1}(\mu^{\theta}(\lambda))P_{N-1}(\mu^{-\theta}(\lambda))} = e^\gamma.
    \end{equation}
    Namely, $\lambda\in \C$ is an eigenvalue of $\cg$ if and only if it satisfies \eqref{eq:evalconstraint}. Moreover, its corresponding eigenspace is always one-dimensional.
\end{theorem}
\begin{proof}
    We will prove the theorem by showing that $(\cg-\lambda I)\bm v=\bm 0$. We consider the equation
    \begin{equation}\label{eq: eigenvectors C setp 1 base}
        \left(\begin{array}{cccccccc}
                e^{\i\theta}(\alpha+\beta) & e^{\i\theta}\eta           &                  &                           &                            &                            &                                      \\
                e^{\i\theta}\beta                  & e^{\i\theta}\alpha & e^{\i\theta}\eta  &                           &                            &                            &                                      \\
                                                  & \ddots                    & \ddots           & \ddots                    &                            &                            &                                      \\
                                                  &                           & e^{\i\theta}\beta & e^{\i\theta}\alpha & e^{\i\theta}\eta            &                            &                                      \\
                                                  &                           &                  & e^{-\i\theta}\eta          & e^{-\i\theta}\alpha & e^{-\i\theta}\beta          &                                      \\
                                                  &                           &                  &                           & \ddots                     & \ddots                     & \ddots                             & \\
                                                  &                           &                  &                           & e^{-\i\theta}\eta           & e^{-\i\theta}\alpha & e^{-\i\theta}\beta                    \\
                                                  &                           &                  &                           &                            & e^{-\i\theta}\eta           & e^{-\i\theta}(\alpha+\beta)   \\
            \end{array}\right)\begin{pmatrix}
            \bm x \\
            \bm y
        \end{pmatrix}= \lambda\begin{pmatrix}
            \bm x \\
            \bm y
        \end{pmatrix}.
    \end{equation}
    We write $\bm x$ as
    \begin{align*}
        \bm x=\left(x_0, \sqrt{\frac{\beta}{\eta}}x_1, \left(\sqrt{\frac{\beta}{\eta}}\right)^2 x_2,  \left(\sqrt{\frac{\beta}{\eta}}\right)^{3} x_3,  \cdots,  \left(\sqrt{\frac{\beta}{\eta}}\right)^{N-1} x_{N-1} \right)^{\top},
    \end{align*}
    and $\bm y$ as
    \[
        \bm y = \left(\left(\sqrt{\frac{\beta}{\eta}}\right)^{N-1} y_{N-1}, \cdots,  \left(\sqrt{\frac{\beta}{\eta}}\right)^{3} y_3, \left(\sqrt{\frac{\beta}{\eta}}\right)^2 y_2, \sqrt{\frac{\beta}{\eta}}y_1, y_0 \right)^{\top}.
    \]
    
    From the first row in (\ref{eq: eigenvectors C setp 1 base}), we can choose
    \[
        x_0 = 1, \quad x_1 = -\frac{\alpha+\beta - e^{-\i\theta}\lambda}{\sqrt{\beta \eta}}.
    \]
    For the second to the $(N-1)$-th row in (\ref{eq: eigenvectors C setp 1 base}), we have
    \[
        \beta \left(\sqrt{\frac{\beta}{\eta}}\right)^j x_j + (\alpha-e^{-\i\theta}\lambda) \left(\sqrt{\frac{\beta}{\eta}}\right)^{j+1}x_{j+1} +\eta \left(\sqrt{\frac{\beta}{\eta}}\right)^{j+2}x_{j+2} =0, \quad j=0,\cdots, N-3.
    \]
    This gives
    \begin{equation}\label{equ:proofeigenrepresent1}
        x_{j+2} = \frac{-(\alpha-e^{-\i\theta}\lambda)}{\sqrt{\beta \eta}}x_{j+1}- x_{j}, \quad j =0, 1, \cdots, N-3.
    \end{equation}
    Thus
    \[
        x_j = P_j(\mu^{\theta}(\lambda)), \quad j=0,\cdots, N-1.
    \]

    Now we consider the last $N$ rows.
    As we have chosen $x_0=1$, $y_0$ should be a constant $C$. Then by the last row in (\ref{eq: eigenvectors C setp 1 base}), we have
    \[
        y_0 = C, \quad y_1 = -\frac{\alpha+\beta - e^{\i\theta}\lambda}{\sqrt{\beta \eta}}C.
    \]
    By the $(2N-1)$-th to the $(N+2)$-th row in (\ref{eq: eigenvectors C setp 1 base}), we have
    \[
        y_{j+2} = \frac{-(\alpha-e^{\i\theta}\lambda)}{\sqrt{\beta \eta}}y_{j+1}- y_{j}, \quad j =0, 1, \cdots, N-3.
    \]
    Thus
    \[
        y_j = C P_j(\mu^{-\theta}(\lambda))\quad j=0,\cdots, N-1.
    \]
    By the above representations, from the $N$-th and $N+1$-th rows in (\ref{eq: eigenvectors C setp 1 base}), we have the equation
    \begin{equation}\label{equ:thetaeigenvalueandcconstraint1}
        \begin{aligned}
             & \left(\sqrt{\frac{\beta}{\eta}}\right)^{N}P_N(\mu^{\theta}(\lambda)) = C\left(\sqrt{\frac{\beta}{\eta}}\right)^{N-1}P_{N-1}(\mu^{-\theta}(\lambda)), \\
             & \left(\sqrt{\frac{\beta}{\eta}}\right)^{N-1}P_{N-1}(\mu^{\theta}(\lambda)) = C\left(\sqrt{\frac{\beta}{\eta}}\right)^{N}P_N(\mu^{-\theta}(\lambda)),
        \end{aligned}
    \end{equation}
    where we have used the $N$-th row and relation (\ref{equ:proofeigenrepresent1}) for $j=N-2$ to compute the first element of $\bm y$ (the treatment to the $N+1$ row is similar).

    Based on (\ref{equ:alphabetagamma}), we can now replace all the $\frac{\beta}{\eta}$ above by $e^{-\gamma}$. Then solving the above equations for $C$ yields (\ref{eq:cdef2}) and (\ref{eq:evalconstraint}).
    As we did not introduce any additional constraints when constructing the eigenvector form (\ref{eq:evecform}), equation (\ref{eq:evalconstraint}) must be the only condition on the eigenvalues $\lambda\in \C$ of $\cg$ and is thus satisfied if and only if $\lambda$ is an eigenvalue. Furthermore, from equations (\ref{eq:evecform}) and (\ref{eq:cdef2}), we can see that for a given eigenvalue $\lambda$ the corresponding eigenvector is uniquely determined (up to constant factor). Hence, the eigenspace corresponding to $\lambda$ must always be one-dimensional.
\end{proof}

These final two facts immediately yield the following characterisation for the exceptional points of $\cg$.
\begin{corollary}\label{cor:epnondestinct}
    An exceptional point occurs when \eqref{eq:evalconstraint} has less than $2N$ distinct solutions.
\end{corollary}

We aim to use this Corollary to show that for a given $\theta>0$, any nonreal eigenvalue $\lambda\in \C\setminus\R$ of $\cg$ must have passed through an exceptional point. However, we must first formalise the notion of an eigenvalue \enquote{having passed through} an exceptional point. To that end, we would like to associate each eigenvalue $\lambda_i$ of $\cg$ with some corresponding continuous path $\lambda_i(\theta)$ such that $\lambda_i(\theta)$ is an eigenvalue of $\cg$ for all values of $\theta$ and all $i=1,\dots, 2N$.
However, precisely because exceptional points occur, we cannot choose these paths in a canonical fashion, as at these exceptional points, two eigenvalue paths $\lambda_i(\theta)$ and $\lambda_j(\theta)$, $i\neq j$, meet and cannot be distinguished. 
What we can do however, is the following: Let $0<\theta'<\frac{\pi}{2}$ be fixed. For any simple eigenvalue $\lambda$ of $\mc{C}^{\theta',\gamma}$, we can then define the maximal unique continuous  \emph{eigenvalue path} $\lambda : \theta\in [\theta_0,\theta']\to \C$ such that $\lambda(\theta)$ is always an eigenvalue of $\cg$ for any $\theta\in [\theta_0,\theta']$ and $\lambda(\theta') = \lambda$. $\theta_0$ is chosen to be either the largest $\theta<\theta'$ such that $\lambda(\theta)$ is an exceptional point, or zero - whichever is greater. 

For some $0<\theta'<\frac{\pi}{2}$ and  $\lambda\in\C$ eigenvalue of $\mc{C}^{\theta',\gamma}$, we can then say that $\lambda$ \emph{has passed through an exceptional point} if any only if $\theta_0$ is greater than zero. Note also that because $C^\gamma$ is diagonalisable, $\lambda(0)$ is never an exceptional point. 

We can now state the following result. 
\begin{corollary}\label{cor:eppassthrough}
    Let $0<\theta'<\frac{\pi}{2}$ and $\lambda\in\C$ an eigenvalue of $\cg$. If $\lambda$ is in $\C\setminus \R$, it must have passed through an exceptional point.
\end{corollary}
\begin{proof}
    We consider the eigenvalue path $\lambda : \theta\in [\theta_0,\theta']\to \C$ as above and aim to prove that $\theta_0 >0$. We assume by contradiction that $\theta_0=0$. Because $C^\gamma$ has real spectrum we must have $\lambda(0)\in \R$,  and there must exist some largest $\theta_r$ such that $\lambda([0,\theta_r])\subset \R$. Because $\lambda(\theta')\in \C\setminus\R$ we know that $\theta_r < \theta'$. But now because of the conjugation symmetry of the spectrum of $\cg$, $\lambda(\theta_r)$ must be a double eigenvalue and thus an exceptional point by the previous corollary. Because $\lambda(0)$ cannot be an exceptional point we must have $0=\theta_0<\theta_r<\theta'$, which yields a contradiction. Therefore, we must at least have $\theta_0=\theta_r>0$ and $\lambda$ must have passed through the exceptional point $\lambda(\theta_r)$.
\end{proof}
This corollary is very useful because it allows us to prove the existence of exceptional points merely from the fact that some eigenvalues are nonreal.

The final result of this section characterises the relation between the factors $C_\lambda$ and $C_{\overline{\lambda}}$ for complex conjugate pairs $\lambda, \overline{\lambda}$ of $\cg$.
\begin{corollary}\label{prop:cconj}
    Let $\lambda, \overline{\lambda}, \in \C$ be a pair of eigenvalues of $\cg$ and $C_\lambda, C_{\overline{\lambda}}$ the corresponding factors as defined in (\ref{eq:cdef2}). Then, we have
    \begin{equation}
        \overline{C_\lambda}C_{\overline{\lambda}}=1.
    \end{equation}
    In particular, we have $\lvert C_\lambda \rvert =1$ for $\lambda\in \R$.
\end{corollary}
\begin{proof}
    Because $P_n$ has real coefficients, we have
    \[
        \overline{P_n(\mu^\theta(\lambda))} = P_n(\mu^{-\theta}(\overline{\lambda})).
    \]
    This fact, together with (\ref{eq:cdef2}),  yields the desired result.
\end{proof}

\section{Eigenvalue of the generalised gauge capacitance matrix}\label{sec:eva}
In this section, we will study the eigenvalues of the generalised gauge capacitance matrix $\cg$ for the system described in \cref{fig:setting}. In \cref{ssec:evacoupled}, we prove that for a small $\theta$ all the eigenvalues of $\cg$ are real, which corresponds to the coupled regime. In \cref{ssec:epexistence} and \cref{ssec:epdensity}, we prove existence and density of exceptional points. Finally, in \cref{ssec:evalocations} we will approximate the locations of eigenvalues. Understanding the movement of eigenvalues will prove to be a crucial prerequisite to understand the decoupling behaviour of the eigenvectors in the next section.

\subsection{Coupled regime}\label{ssec:evacoupled}
This subsection is dedicated to the case where all eigenvalues of $\cg$ are real. We will show that for small $\theta$ the eigenvalues behave similarly to the ones of the gauge capacitance matrix $\capmat^{\theta=0,\gamma}$. Consequently, as will be shown in \cref{sec:eves}, also the eigenvectors of $\cg$ will have a similar form to the ones of $\capmat^{\theta=0,\gamma}$.
\begin{proposition}
    For any $N\in \N$ and $\gamma>0$ there exists a $\varepsilon>0$ such that for $0\leq \theta<\varepsilon$ all the eigenvalues of $\cg$ are real. For a real eigenvalue $\lambda$ of $\cg$ the eigenvector $\bm = (\bm x, \bm y)^\top$ decomposed as in \cref{thm:spectrarepresentation1}, has the following symmetry:
    \begin{equation}
        \bm y = e^{\i\phi}P\overline{\bm x}
    \end{equation}
    for some $\phi \in [0,2\pi)$. In particular we have $\abs{\bm x^{(j)}} = \abs{\bm y^{(2N + 1 - j)}}$ for $j=1,\dots N$.
\end{proposition}
\begin{proof}
    For $\theta=0$, we have $\cg = \cm$ which is quasi-Hermitian (see \cref{sec:matrixsymmetries}) and thus diagonalisable with real spectrum. From \cref{thm:spectrarepresentation1}, we know that the eigenspace for any eigenvalue is one-dimensional. Therefore, $\cm$ must have $2N$ distinct eigenvalues to be diagonalisable.

    Now, because the eigenvalues of a matrix depend continuously on its entries, the map $\theta \mapsto \sigma(\cg)$ must be continuous and there exists some $\varepsilon>0$ such that the eigenvalues remain distinct for $0\leq \theta <\varepsilon$. Because $\cg$ is pseudo-Hermitian (see \cref{def:pseudoherm}), its spectrum must be invariant under complex conjugation and real eigenvalues can only become complex pairwise, after meeting on the real line. Thus, for $\theta$ small enough, no two real eigenvalues of $\cg$ could have met and become complex, ensuring that $\sigma(\cg)\subset \R$.

    The last part follows by the same argument as \cref{prop:cconj} together with $\lambda\in \R$.
\end{proof}

\subsection{Existence of exceptional points}\label{ssec:epexistence}
In this subsection, our aim is to show that regardless of $N$ and $\gamma$, all eigenvalues must pass through an exceptional point as $\theta$ is increased from $0$ to $\frac{\pi}{2}$. 

\begin{theorem}\label{thm::all_eigs_excpetional_point}
Let $\gamma>0$ and $N\in\N$. Then, all but two eigenvalues $\lambda\in \C$ of $\cg$ for $\theta=\frac{\pi}{2}$ must have passed through an exceptional point. The remaining two eigenvalues experience an exceptional point at $\theta=\frac{\pi}{2}$.
\end{theorem}

In line with \cref{cor:eppassthrough}, in order to show that an eigenvalue $\lambda$ of $\cg$ for $\theta=\frac{\pi}{2}$ has passed through an exceptional point, it is sufficient to show that $\lambda$ lies in $\C\setminus\R$. Indeed, for $\theta=\frac{\pi}{2}$ we will show that the spectrum of $\cg$ lies entirely on the imaginary axis. All nonzero eigenvalues must thus have gone through an exceptional point. Two eigenvalues will turn out to be zero, yielding another exceptional point exactly at $\theta=\frac{\pi}{2}$.

We will proceed by giving a useful characterisation of the purely imaginary eigenvalues of $\cg$ for $\theta=\frac{\pi}{2}$, which is based on the characteristic equation \cref{eq:evalconstraint}:

\begin{lemma}
Let $\gamma>0$ be fixed and $\theta=\frac{\pi}{2}$. Then, $\widetilde{\lambda}\in i\R$ is a purely imaginary eigenvalue of $\cg$ if and only if 
\begin{equation}\label{eq:imagconstraint}
\mc{S}(\lambda) \coloneqq P_N(\mu(\lambda))P_N(\mu(-\lambda))-e^\gamma P_{N-1}(\mu(\lambda))P_{N-1}(\mu(-\lambda))=0,
\end{equation}
where
\[
    i\lambda = \widetilde{\lambda} \quad \text{and} \quad \mu(\lambda) = \frac{\lambda-\alpha}{2\sqrt{\beta\eta}}.
\]
\end{lemma}
\begin{proof}
    Suppose that $\widetilde{\lambda}$ is an eigenvalue  of $\cg$ for $\theta=\frac{\pi}{2}$. The result then follows immediately from the characteristic equation \cref{eq:evalconstraint} and realising that $\mu^{\pm\frac{\pi}{2}}(\widetilde{\lambda}) = \mu(\pm\lambda)$.
\end{proof}
Thus, every real zero $\lambda$ of (\ref{eq:imagconstraint}) corresponds to a purely imaginary eigenvalue $i\lambda$ of $\mathcal C^{\frac{\pi}{2},\gamma}$. Note that because $\mc{S}(\lambda)$ is invariant under $\lambda \mapsto -\lambda$, it must be even and its zeros must be symmetric about the origin.

In \cref{sec:poly_inter} we prove the following result:
\begin{proposition}\label{prop:Szeros}
    $\mc{S}(\lambda)$ has exactly $2N-2$ distinct real zeros and a double zero $\lambda=0$.
\end{proposition}
The main idea of the proof is to exploit the heavily interlaced nature of the Chebyshev polynomials, which will prove to be a robust source of zeros of their composites. This will allow us to guarantee and bound $N$ real zeros for $P_N$ and $N-1$ real zeros for $P_N+P_{N-1}$. The evenness of $\mc{S}(\lambda)$ will then allow us to use these results to guarantee zeros of $\mc{S}$ as well.

We can then combine the arguments of this subsection to prove \cref{thm::all_eigs_excpetional_point}.
\begin{proof}[Proof of \cref{thm::all_eigs_excpetional_point}]
    For $\theta=0$, all the eigenvalues are real and for $\theta=\frac{\pi}{2}$ all but two eigenvalues are purely imaginary by \cref{prop:Szeros}. The two not purely imaginary eigenvalues are both zero, causing an exceptional point by \cref{cor:epnondestinct}. By \cref{cor:eppassthrough}, all the purely imaginary, nonzero eigenvalues must have passed through an exceptional point by $\theta=\frac{\pi}{2}$. Furthermore, these eigenvalues are distinct, which ensures that they passed through that exceptional point \emph{before} $\theta=\frac{\pi}{2}$.
\end{proof}

\subsection{Asymptotic density of exceptional points}\label{ssec:epdensity}

\begin{figure}
    \centering
    \includegraphics[width=0.5\textwidth]{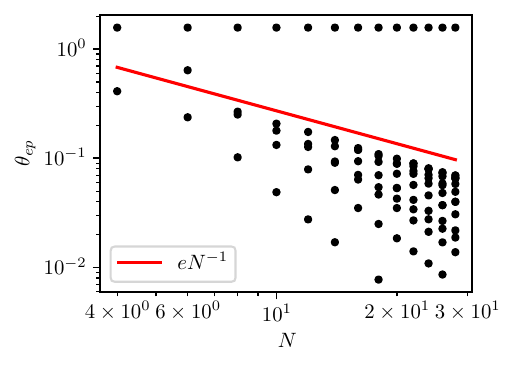}
    \caption{Distribution of the exceptional points for varying $N$. For any $N$, the system exhibits a trivial exceptional point at $\theta=\frac{\pi}{2}$. All other exceptional points concentrate in the interval $[0,e/N]$ and become increasingly dense as $N$ grows.}
    \label{fig:epdependence}
\end{figure}
We are now interested in showing that exceptional points do not only occur (as shown in the previous subsection) but also cluster creating a parameter region with high density of such points. Many of the results developed in this subsection will also be used in \cref{ssec:evalocations} and \cref{sec:eves} as they enable the asymptotic characterisation of the eigenvalue locations and eigenvector growth.

The main aim of this subsection will be to prove the following result.
\begin{theorem}\label{thm:asymptotic density}
    Let $0<\theta<\frac{\pi}{2}$ and $\gamma>0$ be fixed. Then, there exists an $N_0\in \N$ such that for every $N\geq N_0$, the corresponding $\cg$ has exactly two real eigenvalues.
\end{theorem}
By \cref{cor:eppassthrough}, this ensures that all other $2N-2$ eigenvalues in $\C\setminus\R$ must have already passed through an exceptional point before $\theta$.

 We start by stating a helpful reformulation of the characterisation \eqref{eq:evalconstraint} for real eigenvalues.

\begin{proposition}\label{prop: real eig means on level set of ratio of poly}
    $\lambda\in \R$ is a real eigenvalue of $\cg$ if and only if
    \begin{equation}\label{eq:realconstraint}
        \abs{\frac{P_N(\mu^\theta(\lambda))}{P_{N-1}(\mu^\theta(\lambda))}} = e^\frac{\gamma}{2}.
    \end{equation}
\end{proposition}
\begin{proof}
    Because $P_N$ has real coefficients and $\lambda\in \R$ is real, we have 
    $$ 
    P_N(\mu^{-\theta}(\lambda)) = P_N(\overline{\mu^{\theta}(\lambda)}) = \overline{P_N(\mu^\theta(\lambda))}.
    $$ 
    Plugging this into the characteristic equation (\ref{eq:evalconstraint}) yields the desired result.
\end{proof}

The transformation $\mu^\theta(\lambda)= e^{-\i\theta}\lambda\frac{1}{\gamma}\sinh\frac{\gamma}{2}-\cosh\frac{\gamma}{2}$ maps the real line $\R$ onto a line $\mu^\theta(\R)$ in $\C$ rotated by $-\theta$ around the point $-\cosh \frac{\gamma}{2}$. This provides a very geometric view of the zeros of equation (\ref{eq:realconstraint}). In fact, the real spectrum of $\cg$ corresponds to intersections of the line $\mu^\theta(\R)$ and a level set of $\mu \mapsto \abs{\pnr{\mu}}$:
\[
    \sigma(\cg) \cap \R = \mu^\theta(\R) \cap \left\{\mu\in \C :\, \abs{\pnr{\mu}} = e^\frac{\gamma}{2}\right\}.
\]
\begin{figure}[h]
    \centering
    \begin{subfigure}[t]{0.48\textwidth}
        \centering
        \includegraphics[height=0.8\textwidth]{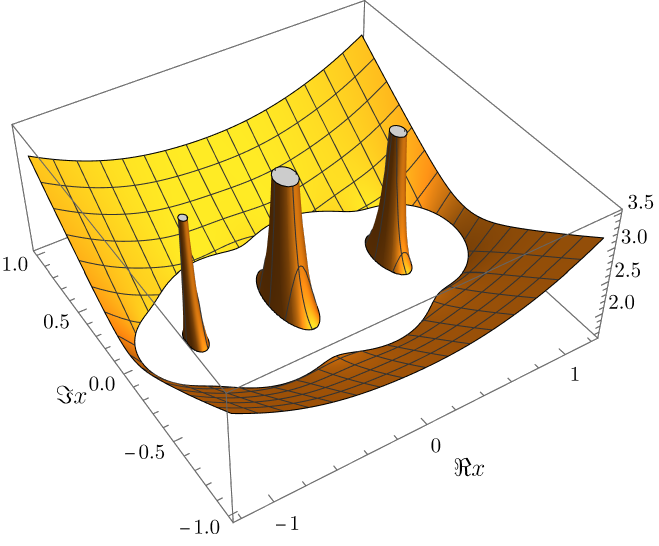}
        \caption{Graph of $\C\ni\mu\mapsto\abs{\pnr{\mu}}$ cut off by the plane $\C\times\{e^\frac{\gamma}{2}\}$.}
        \label{fig: surfaceplot intersect plane}
    \end{subfigure}
    \hfill
    \begin{subfigure}[t]{0.48\textwidth}
        \centering
        \includegraphics[height=0.8\textwidth]{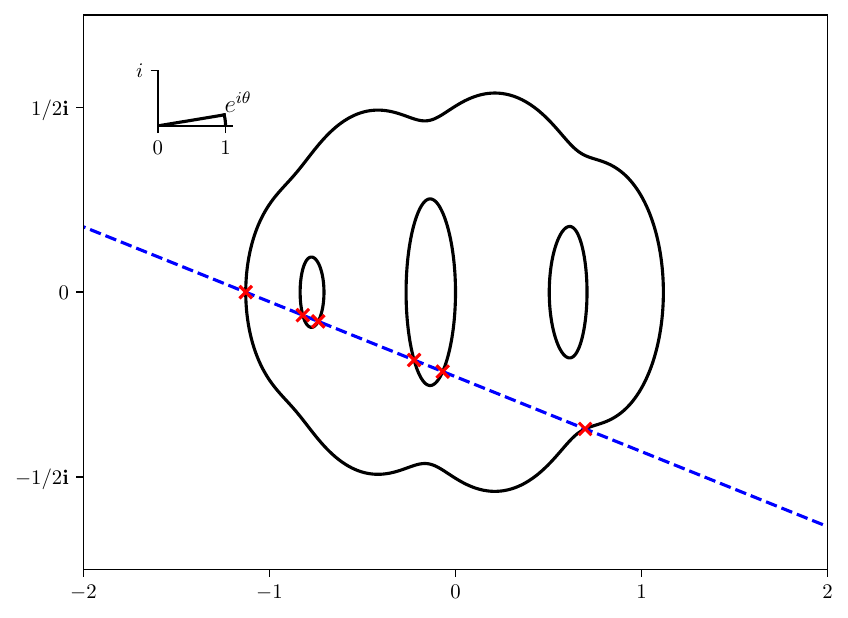}
        \caption{Intersection of $\mu^\theta(\R)$ (in blue dashed) with the level set $\{\mu\in \C :\, \abs{\pnr{\mu}} = e^\frac{\gamma}{2}\}$ (in black). The latter can be seen in \cref{sub@fig: surfaceplot intersect plane}. Intersection points are shown in red. The preimage of these are the real eigenvalues of $\cg$.}
        \label{fig: intersection of levelset with mu}
    \end{subfigure}
    \caption{Geometrical interpretation of the eigenvalues of $\cg$ as given by \cref{prop: real eig means on level set of ratio of poly}. In this view, we can also clearly see the exceptional points, where two real eigenvalues meet and become complex. Namely, this happens exactly when $\mu^\theta(\R)$ goes from passing through one of the inner regions in (B) to moving past them and two red crosses meet.}
    \label{fig: real eigenvalues are points of level set}
\end{figure}
Thus, in order to understand the real eigenvalues of $\cg$, it is crucial to understand the level sets of $\abs{\pnr{\mu}}$. We begin by recalling a well-known equivalent definition of the Chebyshev polynomials:
\[
    U_n(\mu) = \frac{a(\mu)^{n+1}-a(\mu)^{-(n+1)}}{2\sqrt{\mu+1}\sqrt{\mu-1}},
\] where $a(\mu)=\mu + \sqrt{\mu+1}\sqrt{\mu-1}$.

The following two lemmas allow us to characterise the map $a(\mu):\C \to \C$ in terms of its inverse as well as the convergence of $\frac{P_N}{P_{N-1}}$ to $a$ as $N\to \infty$. For the sake of brevity the proofs of these results have been moved to \cref{sec:technical_proofs}.
\begin{lemma}\label{lem:acharact}
    The map $a:\C\setminus[-1,1] \to \{z\in \C \mid \abs z > 1\}$ is a bijective holomorphic map with holomorphic inverse given by
    \begin{equation}\begin{aligned}
            \label{eq:ainv}
            a^{-1}:\{z\in \C \mid \abs z > 1\} & \to \C \setminus [-1,1]                                                                            \\
            z = re^{\i\varphi}                  & \mapsto \frac{1}{2}(z+\frac{1}{z}) = \frac{r^2+1}{2r}\cos \varphi + i\frac{r^2-1}{2r}\sin \varphi.
        \end{aligned}\end{equation}
    
    Although $a$ can be defined on all of $\C$, it fails to be regular at $[-1,1]$. This region is characterised by $a^{-1}(\{z\in \C \mid \abs z = 1\}) = [-1,1]$.
    In particular, the level sets of $a$ are empty for $\abs a <1$, ellipses for $\abs a > 1$ and a line segment for $\abs a=1$.
\end{lemma}

The following result allows us to approximate $\pnr{\mu}$ in terms of $a(\mu)$ with an asymptotically small error as $N\to\infty$.
\begin{lemma}\label{lem:aconv}
    Let $\mu \in \C\setminus[-1,1]$. We have
    \[
        \left|\frac{P_n(\mu)}{P_{n-1}(\mu)}-a(\mu)\right|\leq \abs{a(\mu)}^{-2n+2}\left(2\frac{
            1+e^{-\frac{\gamma}{2}}
        }
        {
            1-\abs{a(\mu)}^{-1}e^{-\frac{\gamma}{2}}-2\abs{a(\mu)}^{-2n+1}
        }\right)
    \]
    for all $n\in \N$ large enough such that $\abs{a(\mu)}^{-2n+2}<\frac{e^\gamma}{2}$. In particular,
    \[
        \frac{P_n(\mu)}{P_{n-1}(\mu)}\stackrel{unif.}{\longrightarrow} a(\mu)
    \] as $n\to \infty$ outside any $\varepsilon$-neighbourhood of $[-1,1]$, i.e. $B_\varepsilon([-1,1])  \coloneqq \{z\in\C: \exists y \in [-1,1]: \vert z - y \vert < \varepsilon\}$.
\end{lemma}

By an analogous argument, we can find that 
$$\left(\frac{P_n}{P_{n-1}}\right)'(\mu)\stackrel{\text{unif.}}{\longrightarrow} a'(\mu) \text{ as } n\to \infty$$
away from $[-1,1]$.

Now that we have a solid understanding of the properties of $a$ and the convergence of $\frac{P_n(\mu)}{P_{n-1}(\mu)}$ to $a$, we can return to proving the matter at hand. 
We are looking for solutions $\lambda\in \R$ of $\abs{\frac{P_N(\mu^\theta(\lambda))}{P_{N-1}(\mu^\theta(\lambda))}} = e^\frac{\gamma}{2}$, since they correspond to the real eigenvalues of $\cg$.
To simplify notation we introduce
\[
    F^\theta_n(\lambda) \coloneqq \abs{\frac{P_n(\mu^\theta(\lambda))}{P_{n-1}(\mu^\theta(\lambda))}}-e^\frac{\gamma}{2} \quad \text{and} \quad F^\theta_\infty(\lambda) \coloneqq \abs{a(\mu^\theta(\lambda))}-e^\frac{\gamma}{2}.
\]
Note that 
In this notation, the results of the previous two lemmas can be summarised and extended as follows.
\begin{proposition}\label{prop: summary properties of F_n}
    For a given $0<\theta<\frac{\pi}{2}$ and $\gamma>0$, we have
    \begin{enumerate}
        \item[(i)] $F_n^\theta(0) = 0$ for all $n\in \N\cup \{\infty\}$;
        \item[(ii)] $F_n^\theta(\lambda) \to \infty$ as $\abs \lambda \to \infty$ for all $n\in \N\cup \{\infty\}$;
        \item[(iii)] $F_n^\theta\stackrel{\text{unif.}}{\longrightarrow}F_\infty^\theta$ and $(F_n^\theta)'\stackrel{\text{unif.}}{\longrightarrow} (F_\infty^\theta)'$ as $n \to \infty$;
        \item[(iv)] $(F_\infty^\theta)^{-1}(0) = \{0,p\}$ for some $0<p\in \R$ and $(F_\infty^\theta)'(0) < 0 < (F_\infty^\theta)'(p)$.
    \end{enumerate}
\end{proposition}
\begin{proof}
    (i) follows from the fact that $0$ is an eigenvalue of $\cg$ for $F^\theta_n$ and a 
    straightforward calculation for $F^\theta_\infty$. Similarly, (ii) follows from the fact that $P_n$ has degree $n$ and from \cref{lem:acharact}. (iii) is the consequence of \cref{lem:aconv} together with the fact that the derivative of the absolute value exists and is continuous away from zero. Finally (iv) follows from (i) and the fact that by \cref{lem:acharact} the level sets of $a$ are ellipses, which are convex.
\end{proof}
Armed with these facts, we can now prove that all but two trivial eigenmodes go through exceptional points for arbitrarily small $\theta$'s as $N\to \infty$.

\begin{proof}[Proof of \cref{thm:asymptotic density}]
    In our notation, we are looking for zeros of $F_N^\theta$ and aim to prove that there exist exactly two for $N$ large enough.
    By \cref{prop: summary properties of F_n} (ii), we know that $F_N^\theta$ can have no zeros for $\lambda$ large. We can thus restrict our search to some large, closed and thus compact set $K\subset \R$.
    Using \cref{prop: summary properties of F_n} (iii) and the fact that $a$ is continuously differentiable, we can find an open neighbourhood $U\coloneqq(-\varepsilon,\varepsilon)\cup (p-\varepsilon,p+\varepsilon)\subset K$ of $0$ and $p$ such that $\abs{a'} > c_1 > 0$ on $U$. We can now use the fact that $(F_N^\theta)'\stackrel{\text{unif.}}{\longrightarrow} (F_\infty^\theta)'$ to find a $N_1\in \N$ such that $\abs{(F_N^\theta)'}>\frac{c_1}{2}>0$ on $U$ for all $N\geq N_1$. For such $N$, $F_N^\theta$ thus has at most two zeros in $U$.

    By \cref{prop: summary properties of F_n} (iv), we know that $F_\infty^\theta\neq 0$ outside $U$. Because it is continuous and $K$ is compact there must be some $c_2$ such that $\abs{F_\infty^\theta}>c_2>0$ on $K\setminus U$. By \cref{prop: summary properties of F_n} (iii) we can now find some $N_2$ such that $\abs{F_N^\theta}>\frac{c_2}{2}>0$ on $K\setminus U$. For such $N$, $F^\theta_N$ thus has no zeros in $K\setminus U$.

    Finally, to prove that $F_n^\theta$ actually has two zeros as well we note that by \cref{prop: summary properties of F_n} (iv) there must be some $0<q<p$ such that $F_\infty^\theta(q)<c_3<0$. By \cref{prop: summary properties of F_n} (iii) there must exist some $N_3$ such that $F_N^\theta(q)<\frac{c_3}{2}<0$ for all $N\geq N_3$. \cref{prop: summary properties of F_n} (ii) and the intermediate value theorem then guarantee that $F_N^\theta(q)$ has two zeros.
    We can now pick $N_0 = \max \{N_1,N_2,N_3\}$ and we are done.
\end{proof}

\subsection{Eigenvalue locations}\label{ssec:evalocations}
In this subsection, we aim to understand the position of the eigenvalues in the complex plane. This will prove crucial in understanding the behaviour of the eigenvectors in \cref{sec:eves}. As we will observe, for a fixed $\theta$ and increasing $N$, they move arbitrarily close to the two line segments $(\mu^\theta)^{-1}([-1,1])\cup (\mu^{-\theta})^{-1}([-1,1])$.

The following result holds.
\begin{proposition}\label{prop:evalocation}
    Let $0< \theta < \pi/2$ and $\gamma>0$ be fixed. For any $\varepsilon>0$ small enough, there exists an $N_0\in \N$ such that for any $N\geq N_0$, all but exactly two eigenvalues of $\cg$ lie in an $\varepsilon$-neighbourhood of $K\coloneqq(\mu^\theta)^{-1}([-1,1])\cup (\mu^{-\theta})^{-1}([-1,1])$. Indeed, because of the conjugation symmetry of eigenvalues, we have
    \begin{equation}
        \left|\sigma(\cg)\cap B_\varepsilon((\mu^{\upsigma\cdot\theta})^{-1}([-1,1])) \right| = N-1,
    \end{equation}
    for $\upsigma =\pm 1$.
\end{proposition}
\begin{proof}
    We recall that $\lambda\in \C$ is an eigenvalue of $\cg$ if and only if it solves the characteristic equation \eqref{eq:evalconstraint}:
    \begin{align}
        \frac{P_{N}(\mu^\theta(\lambda))}{P_{N-1}(\mu^\theta(\lambda))}\frac{P_{N}(\mu^{-\theta}(\lambda))}{P_{N-1}(\mu^{-\theta}(\lambda))} = e^\gamma.
    \end{align}
    Using \cref{lem:aconv}, we find that \[
        \frac{P_{N}(\mu^\theta(\lambda))}{P_{N-1}(\mu^\theta(\lambda))}\frac{P_{N}(\mu^{-\theta}(\lambda))}{P_{N-1}(\mu^{-\theta}(\lambda))} \stackrel{unif.}{\longrightarrow} a(\mu^{\theta}(\lambda))a(\mu^{-\theta}(\lambda))
    \] as $N\to \infty$ outside an $\varepsilon$-neighbourhood of $K\coloneqq(\mu^\theta)^{-1}([-1,1])\cup (\mu^{-\theta})^{-1}([-1,1])$. The same holds for the derivative by an analogous argument.

    \textbf{Claim:} $a(\mu^{\theta}(\lambda))a(\mu^{-\theta}(\lambda))=e^\gamma$ has exactly two solutions and they both lie on the real line.
    
    We begin by proving that such solutions must be real, \emph{i.e.}, $\lambda\in \R$. The fact that the product $a(\mu^{\theta}(\lambda))a(\mu^{-\theta}(\lambda))=e^\gamma \in \R$ must be real implies that $a(\mu^{\pm\theta}(\lambda))$ must have arguments differing only in sign. We note that by \cref{lem:acharact} the set of all $\mu\in \C$ with $\Arg a(\mu) = \pm\varphi$ is given by
    \[
        \left\{\frac{r^2+1}{2r}\cos\varphi\pm i\frac{r^2-1}{2r}\sin\varphi\ \middle|\ r\in [1,\infty)\right\}.
    \]
    These sets form the upper and lower part of the right branch of a hyperbola in the right half complex plane, with focal point $1$. The characteristic equation of this hyperbola is
    \[
        \frac{\Re(\mu)^2}{\cos^2 \varphi}-\frac{\Im(\mu)^2}{\sin^2 \varphi} =1.
    \]
    Any solution of $a(\mu^{\theta}(\lambda))a(\mu^{-\theta}(\lambda))=e^\gamma$ must thus have $\mu_1\coloneqq\mu^{\theta}(\lambda)$ and $\mu_2\coloneqq\mu^{-\theta}(\lambda)$ on opposite branches of this hyperbola. We assume without loss of generality that $\mu_1$ is in the upper branch.
    
    By definition we must also have $(\mu^\theta)^{-1}(\mu_1) =\lambda=(\mu^{-\theta})^{-1}(\mu_2)$. Note that the inverse of $\mu^\theta$ is given by
    \[
        (\mu^{\theta})^{-1}(\mu) = \frac{e^{\i\theta}\gamma}{\sinh \frac{\gamma}{2}}(\mu + \cosh \frac{\gamma}{2}).
    \]
    The next step in showing the claim is to prove that for $\mu_1,\mu_2$ on opposite branches on the parabola $\mu_2 = \overline{\mu_1}$ must hold. We argue by contraposition and assume that $\mu_2 \neq \overline{\mu_1}$ for some $\mu_1,\mu_2$ as above. Because $\overline{\mu_1}$ is the only other point with the same absolute value as $\mu_1$ on this branch of the hyperbola, we must have $\abs{\mu_1}\neq \abs{\mu_2}$. We assume without loss of generality that $\abs{\mu_1}< \abs{\mu_2}$. Because $\mu_1$ and $\mu_2$ lie on the branch in the right half plane this implies $\abs{\Re(\mu_1)}<\abs{\Re(\mu_2)}$ and $\abs{\Im(\mu_1)}<\abs{\Im(\mu_2)}$. But then also $\abs{\mu_1+\cosh\frac{\gamma}{2}}<\abs{\mu_2+\cosh\frac{\gamma}{2}}$ and thus, $\abs{(\mu^{\theta})^{-1}(\mu_1)}<\abs{(\mu^{-\theta})^{-1}(\mu_2)}$. Therefore, $(\mu^\theta)^{-1}(\mu_1) \neq(\mu^{-\theta})^{-1}(\mu_2)$, proving the contrapositive as desired.

    It remains to show that for $\mu_1$ and $\mu_2=\overline{\mu_1}$ on the hyperbola, $\lambda=(\mu^\theta)^{-1}(\mu_1) =(\mu^{-\theta})^{-1}(\mu_2)$ is real. We note
    \[
        \overline{(\mu^\theta)^{-1}(\mu_1)} = (\mu^{-\theta})^{-1}(\overline{\mu_1}) = (\mu^{-\theta})^{-1}(\mu_2) = (\mu^\theta)^{-1}(\mu_1),
    \] which proves that $\lambda=(\mu^\theta)^{-1}(\mu_1)\in \R$, and the first part of the claim is shown.

    We have thus shown that any solution $\lambda$ of $a(\mu^{\theta}(\lambda))a(\mu^{-\theta}(\lambda))=e^\gamma$ must be real. But for real $\lambda$ this equation simplifies to
    \[
        \abs{a(\mu^{\theta}(\lambda))} = e^\frac{\gamma}{2},
    \]
    and by the previous section, we know that there exist exactly two solutions $\{0,p\}\subset \R$, which concludes the proof of the claim.\\[3mm]
    We now know that $a(\mu^{\theta}(\lambda))a(\mu^{-\theta}(\lambda))=e^\gamma$ has exactly two solutions. Since the left-hand side of the characteristic equation \eqref{eq:evalconstraint} and its derivative converge uniformly to $a(\mu^{\theta}(\lambda))a(\mu^{-\theta}(\lambda))$ outside any small $\varepsilon$-neighbourhood of $K$, we can use a similar argument to the one in the previous section to find that, for $N$ large enough, \eqref{eq:evalconstraint} must have exactly two solutions outside of this neighbourhood.

    Because the equation is equivalent to a polynomial of degree $2N$, it must have exactly $2N$ solutions in total. But because only exactly $2$ of these solutions may lie outside the small $\varepsilon$-neighbourhood of $K$, the remaining $2N-2$ must lie in this neighbourhood, as desired. Because the solutions are invariant under complex conjugation they distribute symmetrically into $N-1$ each in the the upper and lower half of $K$. The proof is then complete. 
\end{proof}
\begin{figure}[ht]
    \centering
    \includegraphics[width=0.5\textwidth]{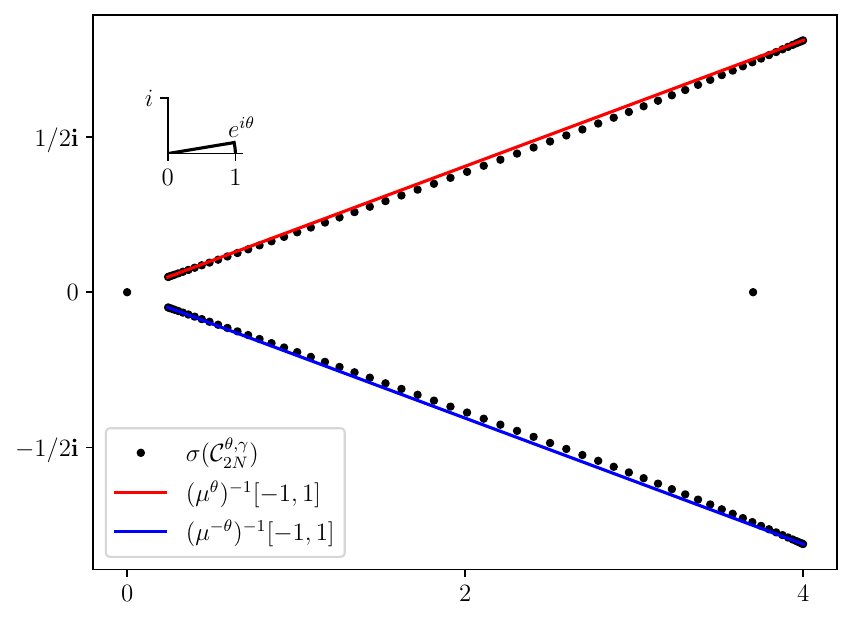}
    \caption{Eigenvalue locations close to the two line segments $(\mu^\theta)^{-1}([-1,1])\cup (\mu^{-\theta})^{-1}([-1,1])$ for $\theta = 0.2$, $\gamma = 1$ and $N=60$.}
    \label{fig:eigenvalues}
\end{figure}

\section{Eigenvectors of the generalised capacitance matrix}\label{sec:eves}
Systems with an imaginary gauge potential are known for the presence of skin effect, \emph{i.e.}, the condensation of the eigenvectors at one edge of the system. This condensation has been shown to be exponential \cite{ammari.barandun.ea2023Perturbed}. The system studied here has a much more peculiar property. The symmetric change of sign in the gauge potential implies that the condensation happens on both edges of the system for small values of $\theta$ or $N$. Nevertheless, the non-Hermiticity introduced by $\theta$ can change this symmetry. The exponential nature of the modes has been shown to be caused by the Freedholm index of the Toeplitz operator associated to the system. The system studied here presented in \cref{fig:setting} does not yield a Toepliz matrix, nevertheless we will show that the same theory can be modified to be used in this situation as well.

\subsection{Exponential decay and growth}
As our matrix $\cg$ is split into two parts by an interface we define the \emph{upper and lower symbols} of $\cg$ as 
\begin{align}\label{eq: symbol def}
    f_\pm^\theta: S^1&\to \C\nonumber\\
    e^{\i\phi} &\mapsto e^{\pm\i\theta}(\beta  e^{\pm\i\phi} + \alpha +  \eta e^{\mp\i\phi}).
\end{align}
We further define the \emph{upper and lower regions of topological convergence} as 
\begin{align}
\label{eq:def_E1_E2}
    E^\theta_\pm = \{z\in\C: \pm \operatorname{wind}(f_\pm^\theta,z) < 0 \},
\end{align}
where $\operatorname{wind}(f_\pm^\theta,z)$ denotes the winding number of $f_\pm^\theta$ around $z$.

These concepts are closely linked to our formalism based on Chebyshev polynomials.
The following result holds. 
\begin{lemma}\label{lem:defineofE2} 
We have
\begin{align*}
    E^\theta_\pm = \{(\mu^{\pm\theta})^{-1}(a^{-1}(re^{\i\phi})) \text{ for } r\in[1,e^{\gamma/2}), \phi\in[0,2\pi)\}.
\end{align*}
\end{lemma}
\begin{proof}
    We will focus on the upper case $E^\theta=E^\theta_+$ and $f^{\theta}=f^{\theta}_+$ as the lower case follows analogously. 
    Algebraic manipulation then yields the following form for the symbol:
    \[
        f^{\theta}(e^{\i\phi}) = e^{\i\theta}\left(-\gamma\coth{\frac{\gamma}{2}}\cos\phi+\i\gamma\sin\phi+\gamma\coth{\frac{\gamma}{2}}\right).
    \]
    Thus $f^{\theta}(e^{\i\phi})$ moves clockwise around an ellipse as $\phi$ goes from $0$ to $2\pi$ and $E^\theta$ must be the interior of this ellipse, by the definition of the winding number.

    We now turn to $(\mu^\theta)^{-1}(a^{-1}(re^{\i\phi}))$ for $r\in[1,e^{\gamma/2}), \phi\in[0,2\pi)$ and aim to show that this also fills out the same ellipse. It is sufficient to show that $g\coloneqq \phi \mapsto (\mu^\theta)^{-1}(a^{-1}(re^{\i\phi}))$ traces the same ellipse as above for $r=e^{\gamma/2}$. This follows from \cref{lem:acharact}, as letting $r$ vary from $1$ to $e^{\gamma/2}$ amounts to filling up the interior of the ellipse drawn out by $g$.
    Algebraic manipulation once again yields
    \[
        g(\phi) = e^{\i\theta}\left(\gamma\coth{\frac{\gamma}{2}}\cos\phi-\i\gamma\sin\phi+\gamma\coth{\frac{\gamma}{2}}\right).
    \]
    This traces the same ellipse as above, with the sole difference that the parametrisation is shifted by $\pi$, concluding the proof.
\end{proof}
By \cref{lem:defineofE2} it makes thus sense to lighten the notation and use $E^\theta=E^\theta_+$ and $E^{-\theta}=E^\theta_-$.

\cref{lem:defineofE2} also justifies calling $E^{\pm\theta}$ \enquote{regions of topological convergence}. Namely, for an eigenpair $(\lambda,\bm v)$ of $\cg$ \cref{thm:spectrarepresentation1} gives the following form for the eigenvector:
\[
    \bm v = (\bm x^{(1)},\dots,\bm x^{(N)},\bm y^{(1)},\dots,\bm y^{(N)})^\top
\]
with $\bm x^{(j)} = (e^{-\frac{\gamma}{2}})^{j-1} P_{j-1}(\mu^\theta(\lambda))$ and $\bm y^{(j)} = (e^{-\frac{\gamma}{2}})^{N-j} P_{N-j}(\mu^{-\theta}(\lambda))$.

The rates of growth for the left and right parts of this eigenvector are then given by
\[
    \frac{\bm x^{(j+1)}}{\bm x^{(j)}} = e^{-\frac{\gamma}{2}}\frac{P_j(\mu^\theta(\lambda))}{P_{j-1}(\mu^\theta(\lambda))}, \quad  \frac{\bm y^{(j+1)}}{\bm y^{(j)}} = e^{\frac{\gamma}{2}}\frac{P_{j-1}(\mu^{-\theta}(\lambda))}{P_{j}(\mu^{-\theta}(\lambda))}.
\]

We focus on the left part we notice its asymptotic growth behavior is determined by whether $\abs{\frac{P_j(\mu^\theta(\lambda))}{P_{j-1}(\mu^\theta(\lambda))}}$ is smaller or larger than $e^{\frac{\gamma}{2}}$. Furthermore, we have $\frac{P_j(\mu^\theta(\lambda))}{P_{j-1}(\mu^\theta(\lambda))} \to a(\mu^\theta(\lambda))$ by \cref{lem:aconv}. $\bm x$ thus decays or grows asymptotically exactly if $\abs{a(\mu^\theta(\lambda))}$ is smaller or larger that $e^{\frac{\gamma}{2}}$ respectively. But by \cref{lem:defineofE2} we can see that $\lambda$ lies in $E^{\theta}$ if and only if $\abs{a(\mu^\theta(\lambda))}<e^{\frac{\gamma}{2}}$, justifying our naming.

For $\bm y$ the rate of growth is exactly the inverse of the rate of growth of $\bm x$ with $\mu^\theta(\lambda)$ replaced by $\mu^{-\theta}(\lambda)$. The above reasoning thus also holds for $\bm y$ with \enquote{growth} and \enquote{decay} as well as the sign of $\theta$ flipped.

 By \cref{prop:evalocation} we know the approximate locations of the eigenvalues of $\cg$. We will now make use of that and the topological convergence to formally prove the above intuition.
 
\begin{theorem}\label{lemma:exponential decay and decoupling}
    Let $0<\theta < \pi/2$, $\gamma>0$ and let $N$ be large enough\footnote{Specifically so that $\varepsilon$ from \cref{prop:evalocation} is smaller than $\sqrt{\alpha+\beta}-\sqrt{\beta+\eta}$ and thus $B_\varepsilon((\mu^{\pm\theta})^{-1}([-1,1]))\subset E^{\pm\theta}$}. Fix, furthermore, a $0<\sigma\ll1$. Consider an eigenpair $(\lambda,\bm v)$ of $\cg$ fulfilling \cref{thm:spectrarepresentation1}. Then, one of the following three cases realises:
    \begin{enumerate}
        \item If $\lambda \in B_\sigma(\partial E^\theta\cup \partial E^{-\theta})\eqqcolon \Theta$, then either $\lambda=0$ and $\bm v = \bm 1$ or no conclusion is made;
        \item If $\lambda \in E^{\theta} \cap E^{-\theta}\setminus \Theta$, then there exist some $B_1,B_2,C_1,C_2>0$ independent of $N$ so that
        \begin{align*}
        \vert \bm{v}^{(j)} \vert < C_1 e^{-B_1\frac{j\gamma}{2}} \quad \text{and}\quad\vert \bm{v}^{(2N+1-j)} \vert < C_2 e^{-B_2\frac{j\gamma}{2}}, 
    \end{align*}
    for $1\leq j\leq N$. In particular, if also $\lambda \in \R$, then $C_1=C_2$ and $B_1=B_2$.
    \item If $\lambda \in E^\theta \triangle E^{-\theta}\setminus \Theta$, then there exits some $B,C>0$ independent of $N$  so that
        \begin{align*}
        \vert \bm{v}^{(j)} \vert < C e^{-B\frac{j\gamma}{2}} \quad \text{if }\lambda\in E^\theta, \\ \quad \vert \bm{v}^{(2N+1-j)} \vert < C e^{-B\frac{j\gamma}{2}} \quad \text{if }\lambda\in E^{-\theta},
    \end{align*}
    for $1\leq j\leq 2N$.
    \end{enumerate}
    In particular, for $\frac{\pi}{4}\leq \theta \leq \frac{\pi}{2}$ case (2) never realises for $N$ large enough.
\end{theorem}

\begin{proof}
    Consider the $\bm x$ part of the eigenvector as described in \cref{thm:spectrarepresentation1}. We have
    \begin{align*}
        \frac{\bm x^{(j+1)}}{\bm x^{(j)}} = e^{-\gamma/2}\frac{P_{j}(\mu^\theta(\lambda))}{P_{j-1}(\mu^\theta(\lambda))}.
    \end{align*}
    Thus, $\bm x$ presents an exponential decay if
    \begin{align}
        \left\vert\frac{P_{j}(\mu^\theta(\lambda))}{P_{j-1}(\mu^\theta(\lambda))}\right\vert <e^{\frac{\gamma}{2}}.
    \end{align}
    The same argument for $\bm y$ shows exponential growth if
    \begin{align}
        \left\vert\frac{P_{j}(\mu^{-\theta}(\lambda))}{P_{j-1}(\mu^{-\theta}(\lambda))}\right\vert <e^{\frac{\gamma}{2}}.
    \end{align}
    We distinguish now two cases and assume without loss of generality that $\lambda$ lies in the upper branch $B_\varepsilon((\mu^{\theta})^{-1}([-1,1]))$, and thus $\lambda\in E^\theta$.\\
    \textbf{Case 1: $\lambda \in E^{-\theta}$}.\\
    Let $0<\varepsilon <\sigma$, then for $\max\{N_1,N_2\}<j<N$ with $N_1$ as in \cref{lem:aconv} and $N_2$ as in \cref{prop:evalocation} the following estimate holds
    \begin{align*}
        \left\vert\frac{P_{j}(\mu^\theta(\lambda))}{P_{j-1}(\mu^\theta(\lambda))}\right\vert &\leq \left\vert\frac{P_{j}(\mu^\theta(\lambda))}{P_{j-1}(\mu^\theta(\lambda))} - a(\mu^\theta(\lambda))\right\vert + \left\vert a(\mu^\theta(\lambda))\right\vert\\
        &\leq \varepsilon + 1 + \varepsilon< e^{\gamma/2}
    \end{align*}
    where the second to last inequality follows from \cref{lem:aconv} and \cref{prop:evalocation}. For $\bm y$ we observe
\begin{align*}
        \left\vert\frac{P_{j}(\mu^{-\theta}(\lambda))}{P_{j-1}(\mu^{-\theta}(\lambda))}\right\vert &\leq \left\vert\frac{P_{j}(\mu^{-\theta}(\lambda))}{P_{j-1}(\mu^{-\theta}(\lambda))} - a(\mu^{-\theta}(\lambda))\right\vert + \left\vert a(\mu^{-\theta}(\lambda))\right\vert\\
        &\leq \varepsilon + (e^{\gamma/2} - \sigma )<e^{\gamma/2},
    \end{align*}
    where the second to last inequality makes additional use of the fact that for $\lambda \in E^{-\theta}\setminus\Theta$ we must have $\abs{a(\mu^{-\theta}(\lambda))}<e^{\gamma/2}-\sigma$. \\
    \textbf{Case 2.} $\lambda \not\in E^{-\theta}$.\\
    Let the constants as in \textbf{Case 1}, but ensure additionally $\varepsilon + (e^{\gamma/2} + \sigma)^{-1} <e^{-\gamma/2}$. Then $\bm x$ still presents exponential decay
    \begin{align*}
        \left\vert\frac{P_{j}(\mu^\theta(\lambda))}{P_{j-1}(\mu^\theta(\lambda))}\right\vert &\leq \left\vert\frac{P_{j}(\mu^\theta(\lambda))}{P_{j-1}(\mu^\theta(\lambda))} - a(\mu^\theta(\lambda))\right\vert + \left\vert a(\mu^\theta(\lambda))\right\vert\\
        &\leq \varepsilon + 1 + \varepsilon <e^{\gamma/2}
    \end{align*}
    because of $\lambda \in E^{\theta}$. On the other side for $\bm y$ we have
    \begin{align*}
        \left\vert\frac{P_{j-1}(\mu^{-\theta}(\lambda))}{P_{j}(\mu^{-\theta}(\lambda))}\right\vert &\leq \left\vert\frac{P_{j-1}(\mu^{-\theta}(\lambda))}{P_{j}(\mu^{-\theta}(\lambda))} - \frac{1}{a(\mu^{-\theta}(\lambda))}\right\vert + \left\vert \frac{1}{a(\mu^{-\theta}(\lambda))}\right\vert\\
        &\leq \varepsilon + (e^{\gamma/2} + \sigma)^{-1} <e^{-\gamma/2}.
    \end{align*}
    Note that \textbf{Case 1} proves point (2) and \textbf{Case 2} proves point (3). It is clear that $\lambda=0$ falls into case (1) and that the corresponding eigenvector is given by $\bm 1$ (see for example \cite{ammari.barandun.ea2024Mathematical}). For $\lambda\in\R$, by \cref{prop:cconj} and the argument in its proof, we conclude that the absolute value of the entries of the eigenvector $\bm v$ must be symmetric with respect to the index $N$.
    
    The last statement of the theorem follows from a geometrical argument. For $\theta=\frac{\pi}{4}$ the major axis of the ellipse $E^{-\theta}$ and the line $(\mu^\theta)^{-1}(\R)$ lay perpendicular to each other and intersect at $0$ while $\vert (\mu^\theta)^{-1}(-1)\vert>0$. Therefore there exists a $\varepsilon$-neighbourhood of $(\mu^\theta)^{-1}([-1,1])$ not intersecting $E^{-\theta}$ and \cref{prop:evalocation} shows the statement.
    \end{proof}

The implications of \cref{lemma:exponential decay and decoupling} are important as it shows that the eigenvectors of $\cg$ are not only exponentially decaying or growing but also that the non-Hermiticity introduced by $\theta$ manifests itself at a macroscopic level as a decoupling of the eigenvectors. While for $\theta=0$ the eigenvectors always present symmetric exponential decay, the non-Hermiticity introduced by $\theta>0$ brings the eigenvalue to eventually migrate to the complex plane and out of one of the two regions $E^\theta$ or $E^{-\theta}$. As a consequence of this, the symmetry is broken. It is also interesting to notice that this process happens pairwise. Since $\cg$ is pseudo-Hermitian, the eigenvalues come in complex conjugated pairs and, as $\theta$ is varying, they meet pairwise at an exceptional point. After the exceptional points, one of the eigenvectors will be decaying while the other will be increasing. The decoupling of the eigenvectors is illustrated in \cref{fig: decoupling of eve}.

From the proof of \cref{lemma:exponential decay and decoupling} we can read out the decay or growth rate of the eigenvectors.

\begin{table}[h]
\begin{tabular}{c|cc|cc|}
\cline{2-5}
                          & \multicolumn{2}{c|}{Upper branch}                                                   & \multicolumn{2}{c|}{Lower branch}                                                 \\ \cline{2-5} 
                          & \multicolumn{1}{c|}{$\lambda\in E^{-\theta}$} & $\lambda\not\in E^{-\theta}$ & \multicolumn{1}{c|}{$\lambda\in E^{\theta}$} & $\lambda\not\in E^{\theta}$ \\ \hline
\multicolumn{1}{|c|}{$\bm x$} & \multicolumn{1}{c|}{$e^{-\gamma/2}$}          & $e^{-\gamma/2}$              & \multicolumn{1}{c|}{$<1$}                    & $>1$                        \\ \hline
\multicolumn{1}{|c|}{$\bm y$} & \multicolumn{1}{c|}{$>1$}                     & $<1$                         & \multicolumn{1}{c|}{$e^{\gamma/2}$}          & $e^{\gamma/2}$              \\ \hline
\end{tabular}
\caption{Approximated decay and growth rate of the left and right part of an eigenvectors of the capacitance matrix depending of the location of the corresponding eigenvalues. Values greater than 1 correspond to growth and lower than 1  correspond to decay. Here upper and lower branch refer respectively to $\lambda\in B_\varepsilon((\mu^{\pm\theta})^{-1}([-1,1]))$ as in \cref{prop:evalocation}.}
\end{table}

\begin{remark}
    A natural question not discussed by \cref{lemma:exponential decay and decoupling} is the behaviour of eigenvectors about their middle, where they cross the interface. To answer this question we consider a nontrivial eigenpair $(\lambda,\bm v)$ of $\cg$. As in \cref{thm:spectrarepresentation1} we decompose $\bm v=(\bm x,\bm y)^\top$ and inspect the ratio $\bm y^{(1)}/\bm x^{(N)}$. Using \cref{thm:spectrarepresentation1} we can find
    \begin{gather*}
        \frac{\bm y^{(1)}}{\bm x^{(N)}} = C\frac{P_{N-1}(\mu^{-\theta}(\lambda))}{P_{N-1}(\mu^{\theta}(\lambda))} = e^{-\frac{\gamma}{2}}\frac{P_{N}(\mu^{\theta}(\lambda))}{P_{N-1}(\mu^{-\theta}(\lambda))}\frac{P_{N-1}(\mu^{-\theta}(\lambda))}{P_{N-1}(\mu^{\theta}(\lambda))} \\= e^{-\frac{\gamma}{2}}\frac{P_{N}(\mu^{\theta}(\lambda))}{P_{N-1}(\mu^{\theta}(\lambda))}.
    \end{gather*}
    This mirrors the relation $\frac{\bm x^{(j+1)}}{\bm x^{(j)}}= e^{-\frac{\gamma}{2}}\frac{P_{j}(\mu^{\theta}(\lambda))}{P_{j-1}(\mu^{\theta}(\lambda))}$ we have for the growth of $\bm x$ and thus the characteristic equation \cref{eq:evalconstraint} takes exactly the form needed to ensure that the growth behaviour of $\bm v$ stays continuous across the interface.
\end{remark}

\begin{figure}[!h]
    \centering
    \begin{subfigure}[t]{0.48\textwidth}
        \centering
        \includegraphics[height=0.55\textwidth]{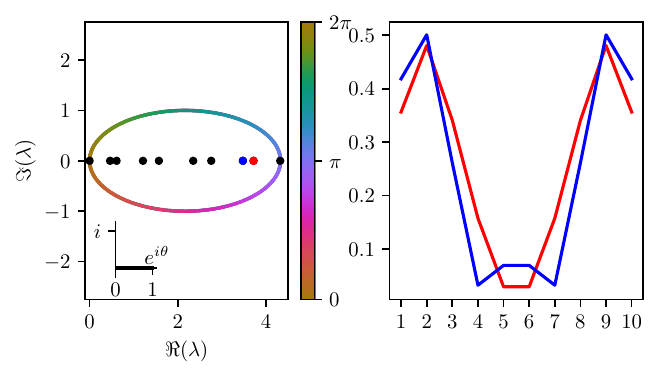}
        \caption{}
    \end{subfigure}
    \hfill
    \begin{subfigure}[t]{0.48\textwidth}
        \centering
        \includegraphics[height=0.55\textwidth]{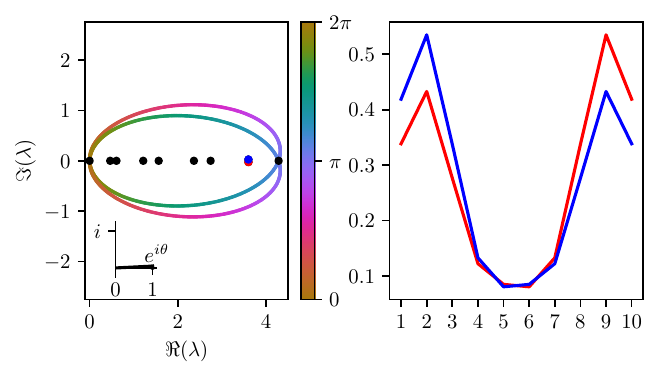}
        \caption{}
    \end{subfigure}\\
    \begin{subfigure}[t]{0.48\textwidth}
        \centering
        \includegraphics[height=0.55\textwidth]{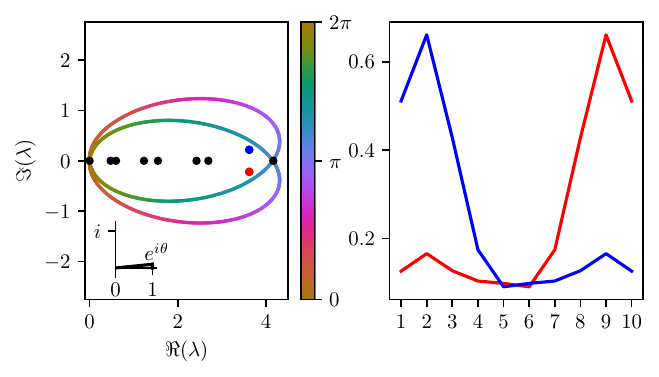}
        \caption{}
    \end{subfigure}
    \hfill
    \begin{subfigure}[t]{0.48\textwidth}
        \centering
        \includegraphics[height=0.55\textwidth]{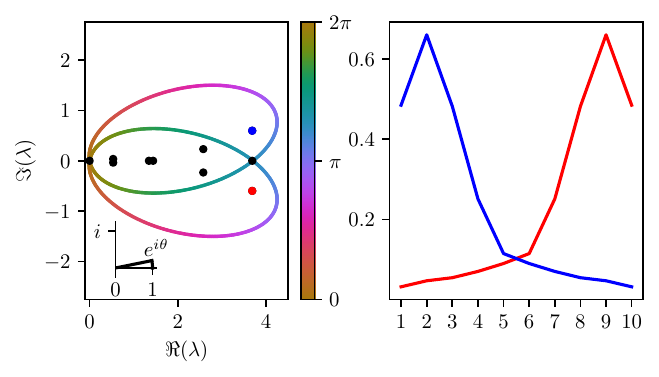}
        \caption{}
    \end{subfigure}
    \caption{Decoupling of the eigenvectors of the gauge capacitance matrix. The macroscopic behaviour of the eigenvectors (exponential decay/growth) is predicted by the location of the eigenvalues in the complex plane with respect to the region of topological convergence defined in \eqref{eq:def_E1_E2} displayed here as trace of \eqref{eq: symbol def}. Looking at the two highlighted eigenvalues (red and blue), Figure (\textsc{A-C}) correspond to point (2) in \cref{lemma:exponential decay and decoupling} while (\textsc{D}) corresponds to point (1).}
    \label{fig: decoupling of eve}
\end{figure}

\subsection{Topological origin}
This section is devoted to illustrating the topological origin of the specific condensation properties of $\mc{C}^{\theta,\gamma}$'s eigenvectors that were shown in Theorem \ref{lemma:exponential decay and decoupling}. Especially, in Theorem \ref{lemma:exponential decay and decoupling} we demonstrate the condensation properties by the specific behaviours of $a(\mu)$, while here we concentrate on the Fredholm index theory of the symbol of the Toeplitz operator, which directly leads to considering $E^{\pm\theta}$ in (\ref{eq:def_E1_E2}) defined by symbols $f^{\theta}_\pm$. Instead of considering the eigenvector of a Toeplitz operator, we analyse the pseudo-eigenvector of the corresponding matrix $\mc{C}^{\theta,\gamma}$ to keep the boundary of the system. We first present the topological origin of the case (2) in Theorem \ref{lemma:exponential decay and decoupling}. 
\begin{theorem}
 Suppose that $\lambda \in E^\theta
 \cap E^{-\theta}$. For some $0<\rho<1$ and sufficiently large $N$, there exists a pseudo-eigenvector $\bm v$ of $\mc{C}^{\theta,\gamma}$ satisfying
 \[
    \frac{\lVert (\mc{C}^{\theta,\gamma}-\lambda)\bm v\rVert}{\lVert \bm v \rVert} \leq \max(C_1, N C_2)\rho^{N-1}
     \]
     such that
    \[
        \begin{cases} \frac{|\bm v_j|}{\max_{i}|\bm v_i|} \leq C_3 N\rho^{j-1}, \quad j =1,\cdots, N, \\
        \frac{|\bm v_j|}{\max_{i}|\bm v_i|} \leq C_3 N\rho^{2N-j}, \quad j =N+1,\cdots, 2N, 
        \end{cases}
     \]
     where $C_1, C_2, C_3$ are constants independent of $N$.
\end{theorem}
\begin{proof}
We first consider the pseudo-eigenvectors of 
\[
 e^{i \theta }T_l \text{ and }  e^{-i \theta }T_r,
\]
where 
\begin{equation}\label{equ:prooftopological0}
T_l = \left(\begin{array}{ccccc}
            \alpha+\beta & \eta   &        &        &             \\
            \beta          & \alpha & \eta   &        &             \\
                           & \ddots & \ddots & \ddots &             \\
                           &        & \beta  & \alpha & \eta        \\
                           &        &        & \beta  & \alpha
        \end{array}\right), \quad T_r = \left(\begin{array}{ccccc}
            \alpha & \beta  &        &        &              \\
            \eta          & \alpha & \beta  &        &              \\
                          & \ddots & \ddots & \ddots &              \\
                          &        & \eta   & \alpha & \beta        \\
                          &        &        & \eta   & \alpha+\beta
        \end{array}\right).
\end{equation}
Since $T_l$ is a Toeplitz matrix with only a perturbation on the first element, by the theory for the pseudo-eigenvector of $T_l$ in \cite{ammari.barandun.ea2024Spectra}, we have that for each $\lambda\in E^{\theta}$, there exist nonzero pseudo-eigenvectors $\bm x$ satisfying
    \[
    \frac{\lVert ( e^{i \theta }T_l-\lambda)\bm x\rVert}{\lVert \bm x \rVert} \leq \max(C_4, N C_5)\rho^{N-1}
     \]
     such that
    \begin{equation}\label{equ:prooftopological1}
         \frac{|\bm x_j|}{\max_{i}|\bm x_i|} \leq C_6 N\rho^{j-1}, \quad j =1,\cdots, N, 
     \end{equation}
     where $C_4, C_5, C_6$ are constants independent of $N$. On the other hand, observing that 
     \[
T_r = P T_l P
     \]
    where $P$ is the anti-diagonal involution. Therefore, by the same theory, for $\lambda\in E^{-\theta}$, there exist nonzero pseudo-eigenvectors $\bm y$ satisfying
    \[
         \frac{\lVert ( e^{-i \theta }T_r-\lambda)\bm y\rVert}{\lVert \bm y \rVert} \leq \max(C_7, N C_8)\rho^{N-1}
     \]
     such that
    \begin{equation}\label{equ:prooftopological2}
         \frac{|\bm y_j|}{\max_{i}|\bm y_i|} \leq C_9 N\rho^{N-j},\quad j =1,\cdots, N,
     \end{equation}
     where $C_7, C_8, C_9$ are constants independent of $N$. Now, we construct the pseudo-eigenvector of $\mc{C}^{\theta,\gamma}$ as $\bm v = \begin{pmatrix}
            \bm x \\
            \bm y
        \end{pmatrix}$. Note that the only difference between 
     \[
     \mc{C}^{\theta,\gamma},\quad \begin{pmatrix}
      e^{i \theta }T_l & \\
      & e^{-i \theta }T_r
     \end{pmatrix}
     \]
     are two elements in the centre of the matrix.  Together with (\ref{equ:prooftopological1}) and (\ref{equ:prooftopological2}), we can conclude that for $\lambda \in E^{\theta}\cap E^{-\theta}$, the pseudo-eigenvector $\bm v$ satisfies that
 \[
    \frac{\lVert (\mc{C}^{\theta,\gamma}-\lambda)\bm v\rVert}{\lVert \bm v \rVert} \leq \max(C_1, N C_2)\rho^{N-1}
     \]
     such that
    \[
        \begin{cases} \frac{|\bm v_j|}{\max_{i}|\bm v_i|} \leq C_3 N\rho^{j-1}, \quad j =1,\cdots, N, \\
        \frac{|\bm v_j|}{\max_{i}|\bm v_i|} \leq C_3 N\rho^{2N-j}, \quad j =N+1,\cdots, 2N, 
        \end{cases}
     \]
     where $C_1, C_2, C_3$ are constants independent of $N$. This completes the proof.
\end{proof}

Now, we present the topological origin of the case (3) in Theorem \ref{lemma:exponential decay and decoupling}. 
\begin{theorem}
Suppose that $\lambda \in E^{\theta} \triangle E^{-\theta}$. For some $0<\rho<1$ and sufficiently large $N$, there exists a pseudo-eigenvector $\bm v$ of $\mc{C}^{\theta,\gamma}$ satisfying
 \[
    \frac{\lVert (\mc{C}^{\theta,\gamma}-\lambda)\bm v\rVert}{\lVert \bm v \rVert} \leq \max(C_1, N C_2)\rho^{N-1}
     \]
     such that
    \[
        \begin{cases} \frac{|\bm v_j|}{\max_{i}|\bm v_i|} \leq C_3 N\rho^{j-1}, \quad j =1,\cdots, 2N, \text{ if $\lambda \in E^{\theta}$, } \\
        \frac{|\bm v_j|}{\max_{i}|\bm v_i|} \leq C_3 N\rho^{2N-j}, \quad j =1,\cdots, 2N, \text{ if $\lambda \in E^{-\theta}$, }
        \end{cases}
     \]
     where $C_1, C_2, C_3$ are constants independent of $N$.
\end{theorem}
\begin{proof} We consider the pseudo-eigenvectors of $e^{i \theta }T_l, e^{-i \theta }T_r$, where $T_l, T_r$ are defined by (\ref{equ:prooftopological0}). For $\lambda \in E^{\theta} \triangle E^{-\theta}$ and $\lambda \in E^{\theta}$, there exist nonzero pseudo-eigenvectors $\bm x$ satisfying
    \[
    \frac{\lVert ( e^{i \theta }T_l-\lambda)\bm x\rVert}{\lVert \bm x \rVert} \leq \max(C_4, N C_5)\rho^{N-1}
     \]
     such that
    \begin{equation}\label{equ:prooftopological3}
         \frac{|\bm x_j|}{\max_{i}|\bm x_i|} \leq C_6 N\rho^{j-1}, \quad j =1,\cdots, N, 
     \end{equation}
     where $C_4, C_5, C_6$ are constants independent of $N$. On the other hand, since $\lambda  \in E^{\theta} \triangle E^{-\theta}$ and $\lambda \not \in E^{-\theta}$, we have 
     \[
    \lambda \in  \{z\in\C:  -\operatorname{wind}(f^{\theta}_-,z) < 0 \}
     \]
     for $f^{\theta}_-$ defined in (\ref{eq: symbol def}). Therefore, by theory in \cite{ammari.barandun.ea2024Spectra} there exist nonzero pseudo-eigenvectors $\bm y$ satisfying  
     \begin{equation}\label{equ:prooftopological4}
    \frac{\lVert \bm y^{\top}( e^{-i \theta }T_l^{\top}-\lambda)\rVert}{\lVert \bm y \rVert} \leq \max(C_7, N C_8)\rho^{N-1}
     \end{equation}
     such that
    \begin{equation}
         \frac{|\bm y_j|}{\max_{i}|\bm y_i|} \leq C_9 N\rho^{j-1}, \quad j =1,\cdots, N, 
     \end{equation}
     where $C_7, C_8, C_9$ are constants independent of $N$. Thus, by (\ref{equ:prooftopological4}), 
     \[
     \frac{\lVert ( e^{-i \theta }T_l-\lambda) \bm y\rVert}{\lVert \bm y \rVert} \leq \max(C_4, N C_5)\rho^{N-1}.
     \]
Now, we construct the pseudo-eigenvector of $\mc{C}^{\theta,\gamma}$ as $\bm v= \begin{pmatrix}
            \bm x \\
            \rho^{N} \bm y
        \end{pmatrix}$. Note that the only difference between 
     \[
     \mc{C}^{\theta,\gamma},\quad \begin{pmatrix}
      e^{i \theta }T_l & \\
      & e^{-i \theta }T_r
     \end{pmatrix}
     \]
     are two elements in the centre of the matrix. By all the above arguments, one can justify that
        \[
    \frac{\lVert (\mc{C}^{\theta,\gamma}-\lambda)\bm v\rVert}{\lVert \bm v \rVert} \leq \max(C_1, N C_2)\rho^{N-1}
     \]
     and 
    \[
        \frac{|\bm v_j|}{\max_{i}|\bm v_i|} \leq C_3 N\rho^{j-1}, \quad j =1,\cdots, 2N, 
    \]
    where $C_1, C_2, C_3$ are constants independent of $N$. This proves the theorem for the case when $\lambda \in E^{\theta} \triangle E^{-\theta}$ and $\lambda \in E^{\theta}$. For the case when $\lambda \in E^{\theta} \triangle E^{-\theta}$ and $\lambda \in E^{-\theta}$, the justification is similar.
\end{proof}

\section{Concluding remarks}
This paper extensively studies both qualitatively and quantitatively parity-time symmetric structures of one-dimensional subwavelength resonators equipped with two kinds of non-Hermiticity: an imaginary gauge potential (leading to a directional decoupling) and complex material parameters (on-site gain and loss). Our results are multifold. First, we have used Chebyshev polynomials to give formulas for the eigenvalues and eigenvectors of the generalised gauge capacitance matrix, which has been shown to approximate the resonance problem at subwavelength scales. Then we have studied the phase change of the spectrum, varying from purely real to complex as a function of the on-site gain and loss parameter $\theta$. Parallel to the spectral change, the eigenvectors go through a decoupling procedure.

Similar systems have been analysed in the quantum mechanical setting \cite{jana.sirota2023Emerging}. The framework presented there differs from ours as no edge effects are considered. Our results can be easily generalised to include these simpler systems by replacing the polynomials $P_n(z)=U_n(z)+e^{\gamma/2}U_{n-1}(z)$ with the Chebyshev polynomials $U_n(z)$. Furthermore, our analysis presents deep insights into both the phenomenological landscape of effects and crucially also into the mathematical foundations of the studied systems and associated physical phenomena.

Our work lends itself to a number of generalisations. On one hand, more exotic $\mc{PT}$-symmetric structures might be analysed. On the other hand, the explicit theory we have developed in the current work relied on the rather simple structure of the generalised gauge capacitance matrix that in one dimension takes a tridiagonal form. Higher-dimensional systems do not enjoy this property, having a dense capacitance matrix \cite{ammari.barandun.ea2023NonHermitian}. Nevertheless, since we were able to relate the decoupling of the eigenvalues to the winding numbers presented in \cref{sec:eves} and since these topological magnitudes can similarly be computed in higher dimensions, we conjecture that very similar results can also be obtained in those setups.

\section*{Acknowledgements}
The authors would like to thank Erik Orvehed Hiltunen and Bryn Davies for insightful discussions. The work of PL was supported by Swiss National Science Foundation grant number 200021--200307. 
\appendix

\section{Matrix symmetries}\label{sec:matrixsymmetries}
This appendix recalls some helpful and well-known results about matrix symmetries \cite[Chapter 2]{barnett2023Locality}.
\begin{definition}[Pseudo-Hermitian]\label{def:pseudoherm}
    A matrix $A\in \C^{n\times n}$ is said to be \emph{pseudo-Hermitian} if there exists a Hermitian matrix $M=M^{*}\in\operatorname{GL}_n(\C)$ such that $MA=A^{*}M$.
\end{definition}

Importantly pseudo-Hermitian matrices can be characterised in numerous ways.
\begin{lemma}\label{thm:pseudoherm}
    For $A\in \C^{n\times n}$ the following properties are equivalent:
    \begin{enumerate}
        \item[(i)] $A$ is pseudo-Hermitian with $A^{*}=MAM\inv$;
        \item[(ii)] $A=G\eta$ where $G = G^*, \eta = \eta^*$ and $\eta \in \operatorname{GL}_n(\C)$ invertible;
        \item[(iii)] $A$ is similar to a matrix with real entries;
        \item[(iv)] $A$ has an involutive antilinear symmetry;
        \item[(v)] $A$ has an invertible antilinear symmetry;
        \item[(vi)] $A$ is similar to $A^*$;
        \item[(vii)] $A$ is weakly pseudo-Hermitian;
        \item[(viii)] $\sigma(A) = \overline{\sigma(A)}$ and $\dim \ker(\lambda\mathbf{1}-A)^l = \dim \ker(\overline{\lambda}\mathbf{1}-A)^l$ for every $l\in \N$ and $\lambda\in \C$.
    \end{enumerate}
\end{lemma}

\begin{corollary}
    The coefficients of the characteristic polynomial of a pseudo-Hermitian matrix are real.
\end{corollary}

\begin{definition}[Quasi-Hermitian]
A matrix $A\in \C^{n\times n}$ is said to be \emph{quasi-Hermitian} if there exists a positive-definite Hermitian matrix $S$ of full rank (called metric operator) such that $SA = A^*S$.
\end{definition}
\begin{corollary}\label{cor:quasiherm} We have
\begin{itemize}
    \item[(i)] $A$ is quasi-Hermitian if and only if it is diagonalisable and has a purely real spectrum;
    \item[(ii)] Every quasi-Hermitian matrix is pseudo-Hermitian. 
\end{itemize}
\end{corollary}

\section{Zero gain and loss case} \label{sec:0gain}
\begin{figure}
    \centering
    \includegraphics[width=0.4\textwidth]{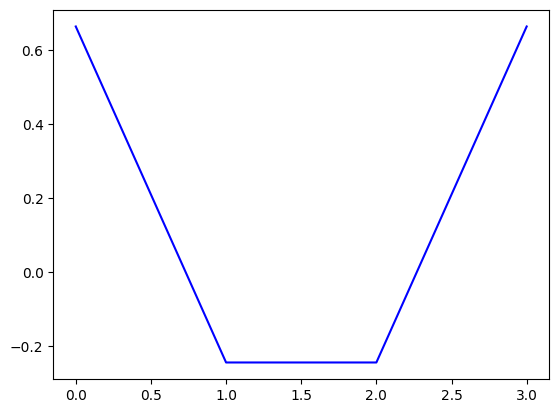}
    \includegraphics[width=0.4\textwidth]{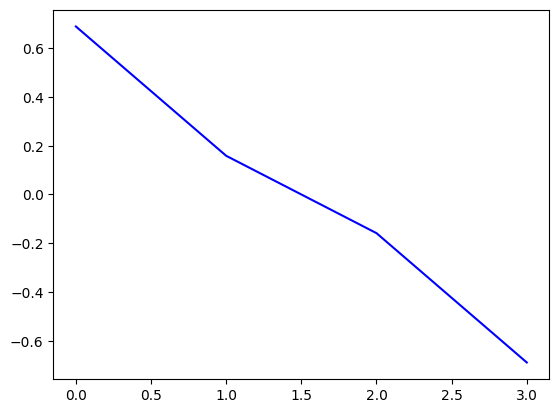}
    \caption{The nontrivial eigenvectors $\widetilde{v}_\pm$ of $C^U_+$ and $C^U_-$ for $N=2$,  respectively (as in \Cref{thm:samespec}). We can see that $\widetilde{v}_+$ is symmetric and $\widetilde{v}_-$ is anti-symmetric.}
    \label{fig:poles}
\end{figure}

In this appendix, we specifically investigate the case $\theta=0$, where there is no gain or loss introduced into the system. We have $\cg = \cm$ and will see that the spectrum of $C^\gamma$ splits evenly into monopole and dipole modes which are either symmetric or antisymmetric about their middle.

We exploit the rich symmetry of $\cm$ to reduce the problem of finding the spectrum of $\cm$ to finding the spectrum of the tridiagonal almost-Toeplitz matrices $C^U_\pm$. This is significantly easier as the spectral theory of tridiagonal Toeplitz matrices with perturbed edges is well-understood (see, for instance, \cite{yueh.cheng2008Explicit}).

The central idea is to investigate the following matrix, which will turn out to have the same eigenvalues and eigenvectors as $\cm$.
\begin{definition}
    We define the \emph{modified capacitance matrix} $\widehat{C}_\pm$ as follows:
    \begin{equation}
        \widehat{C}_\pm =
        \left(\begin{array}{ccccc|ccccc}
                \alpha + \beta & \eta   &        &        &                               &                               &        &        &        &              \\
                \beta          & \alpha & \ddots &        &                               &                               &        &        &        &              \\
                               & \ddots & \ddots &        &                               &                               &        &        &        &              \\
                               &        &        & \alpha & \eta                          &                               &        &        &        &              \\
                               &        &        & \beta  & {\color{red} \alpha \pm \eta} & 0                             &        &        &        &              \\
                \hline
                               &        &        &        & 0                             & {\color{red} \alpha \pm \eta} & \beta                                   \\
                               &        &        &        &                               & \eta                          & \alpha & \ddots                         \\
                               &        &        &        &                               &                               & \ddots & \ddots                         \\
                               &        &        &        &                               &                               &        &        & \alpha & \beta        \\
                               &        &        &        &                               &                               &        &        & \eta   & \alpha+\beta
            \end{array}\right) \in \R^{2N\times 2N}.
    \end{equation}
    We will refer to the top-left block of $\widehat{C}_\pm$ as $\cu$ and the bottom-left block as $\cl$.
\end{definition}

\begin{remark}\label{rmk:cuprop}
    Note that $\cu$ as well as $\cl$ are tridiagonal Toeplitz matrices with perturbed edges.
    We can see that these matrices are similar by
    \[P\cu P = \cl,\]
    which means that they must share the same spectrum. Furthermore, any eigenvector $\bm v$ of $\cu$ directly corresponds to an eigenvector of $\cl$:
    \begin{equation}\label{eq:cuspec}
        \cu \bm v = \lambda \bm v \iff  \cl P\bm v = \lambda P\bm v \quad \text{ for }\lambda \in \R, \bm v\in \R^N.
    \end{equation}

    Furthermore, analogously to \cref{prop:cquasiherm} $\cu,\cl$  are diagonalisable with real spectrum.
\end{remark}

\begin{lemma}\label{lem:modc}
    Let $\widetilde{\bm v}_\pm = (v_1,\dots, v_{2N}) \in \R^{2N}$. We have $v_{N+1} = \pm v_N$ if and only if
    \begin{equation}
        \widehat{C}_\pm \widetilde{\bm v}_\pm = \cm \widetilde{\bm v}_\pm . 
    \end{equation}
\end{lemma}
\begin{proof}
    We will prove that $\ker (\widehat{C}_\pm - \cm) = \{\bm v\in \R^{2N} \mid v_{N+1} = \pm v_N\}$. To that end, we notice that for any $\bm v\in \R^{2N}$
    \[
        (\widehat{C}_\pm - \cm)\bm v =
        \left(\begin{array}{ccc|ccc}
                 &  &         &         &  & \\
                 &  &         &         &  & \\
                 &  & \pm\eta & -\eta   &  & \\
                \hline
                 &  & -\eta   & \pm\eta &  & \\
                 &  &         &         &  & \\
                 &  &         &         &  & \\
            \end{array}\right)\bm v = \begin{pmatrix}
            0                          \\
            \vdots                     \\
            0                          \\
            \pm\eta v_N -\eta v_{N+1}  \\
            -\eta v_N \pm \eta v_{N+1} \\
            0                          \\
            \vdots                     \\
            0
        \end{pmatrix} , 
    \]
    which implies that $\bm v\in \ker (\widehat{C}_\pm - \cm)$ if and only if $v_{N+1} = \pm v_N$.
\end{proof}

\begin{proposition}\label{thm:samespec}
    Let $\bm v_\pm \in \R^N$ and $\lambda \in \R$. We define $\widetilde{\bm v}_\pm \coloneqq (\bm v_\pm,\pm P \bm v_\pm) \in \R^{2N}$.
    Then, the following statements are equivalent:
    \begin{enumerate}[label=\normalfont(\arabic*)]
        \item[(i)] $\bm v_\pm$ is an eigenvector of $\cu$ with eigenvalue $\lambda$;
        \item[(ii)] $\widetilde{\bm v}_\pm$ is an eigenvector of $\widehat{C}_\pm$ with eigenvalue $\lambda$;
        \item[(iii)] $\widetilde{\bm v}_\pm$ is an eigenvector of $\cm$ with eigenvalue $\lambda$.
    \end{enumerate}
\end{proposition}
\begin{proof}
    \enquote{$(i)\iff(ii)$}: $(ii)\implies(i)$ follows directly from the fact that $\widehat{C}_\pm$ is block-diagonal with upper-left block $\cu$. $(i)\implies(ii)$ holds because by Remark \ref{rmk:cuprop}, $\bm v_\pm$ being an eigenvector of $\cu$ associated to the eigenvalue $\lambda$ implies that $P\bm v_\pm$ is an eigenvector of $\cl$ with eigenvalue $\lambda$. This fact remains unchanged if we possibly change the sign of $P\bm v_\pm$. $(ii)$ then follows from the fact that $\widehat{C}_\pm$ is block-diagonal with blocks $\cu$ and $\cl$.

    \enquote{$(ii)\iff(iii)$}: Follows directly from \Cref{lem:modc}.
\end{proof}

\Cref{lem:modc} establishes the connection of the spectrum of $\cu$ with the spectrum of $\cm$. Because both $C^U_+$ and $C^U_-$ are diagonalisable with $N$ eigenvalues each we can find all $2N$ eigenvalues of $\cm$ and obtain the following result.
\begin{corollary} We have
    \begin{equation}
        \sigma(\cm) = \sigma(C^U_+) \sqcup \sigma(C^U_-).
    \end{equation}
\end{corollary}
Hence, we can split the spectrum of $\cm$ into $N$ \emph{monopole modes} which are symmetric and $N$ \emph{dipole modes} which are anti-symmetric about the middle (see Figure \ref{fig:poles}).

We can now use \cite[Theorem 3.1]{yueh.cheng2008Explicit} to explicitly find the monopole eigenvalues and eigenvectors.
\begin{lemma}
    The eigenvalues of $C^U_+$ are $\lambda_1=0$ with corresponding eigenvector $u_1=(1,\dots,1)\in \R^{2N}$ or 
    \begin{equation*}
        \lambda_k = \alpha + 2\sqrt{\beta\eta}\cos\left(\frac{\pi}{N}(k-1)\right), \quad 2\leq k \leq N
    \end{equation*}
    with corresponding eigenvector
    \begin{equation*}
        u^{(j)}_k = \left(\frac{\beta}{\eta}\right)^{\frac{j-1}{2}}\left(\beta \sin \left(\frac{j(k-1)\pi}{N}\right)-\beta\sqrt{\frac{\beta}{\eta}}\sin \left(\frac{(j-1)(k-1)\pi}{N}\right)\right),
    \end{equation*}
    for $2\leq k \leq N$ and $1\leq j\leq N$.
\end{lemma}

\section{Polynomial interlacing}\label{sec:poly_inter}
This appendix complements \cref{ssec:epexistence} providing a proof for the following proposition.
\begin{proposition}
    $\mc{S}(\lambda)$ has exactly $2N-2$ distinct real zeros and a double zero $\lambda=0$.
\end{proposition}

We will exploit the heavily interlaced nature of the Chebyshev polynomials in order to find a robust source of zeros of their composites. This will allow us to guarantee and bound $n$ real zeros for $P_n$ and $n-1$ real zeros for $P_n+P_{n-1}$. The evenness of $\mc{S}(\lambda)$ will then allow us to use these results to guarantee zeros of $\mc{S}$ as well.

We begin with the following facts about such interlaced polynomials.
\begin{definition}
    For a differentiable function $f:[a,b] \to \R$ with a simple zero $x\in \R$, we define the \emph{sign} of $x$ to be positive if $p'(x)>0$ and negative if $p'(x)<0$, and write $\sgn x = \pm 1$.
\end{definition}
\begin{remark}\label{rmk:invariantzeros}
    Multiplying $f$ by some positive differentiable function $g>0$ has no impact on the sign or locations of the zeros of $gf$.

    Furthermore, for a polynomial with positive leading coefficient and only simple zeros, the largest zero has positive sign and every smaller zero has alternating signs.
\end{remark}
\begin{lemma}\label{lem:interlacingfacts}
    Let $p$ and $q$ be polynomials with real coefficients. Then, the following can be said about the zeros of $p-q$:
    \begin{enumerate}
        \item[(i)] Let $x$ and $y$ be zeros of $p$ and $q$ respectively with no other zeros between them. If they have opposing signs, then $p-q$ must have zero between $x$ and $y$.
    \end{enumerate}
    Furthermore, if $\deg p > \deg q$, then we also have the following results:
    \begin{enumerate}[start=2]
        \item[(ii)] Let $x$ and $y$ be the smallest zeros of $p$ and $q$ respectively. If $x<y$ and they have equal sign, then $p-q$ must have a zero $z<x<y$;
        \item[(iii)] Let $x$ and $y$ be the largest zeros of $p$ and $q$ respectively. If $y<x$ and they have equal sign, then $p-q$ must have a zero $y<x<z$.
    \end{enumerate}
\end{lemma}
\begin{proof}
    (i) We assume without loss of generality that $x$ has positive sign and $y$ has negative sign. Because neither $p$ nor $q$ have other zeros between $x$ and $y$, we know that $q(x) > 0$ and $p(y)>0$. Because $p(x)=q(y)=0$, we get $p(x)-q(x) < 0$ and $p(y)-q(y) > 0$. Thus, $p-q$ has a sign flip in $(x,y)$ and by the intermediate value theorem, there must exist a $z\in (x,y)$ such that $p(z)-q(z)=0$.

    For (ii) we assume without loss of generality that the signs of $x$ and $y$ are negative. Since by assumption $y$ is the smallest zero of $q$, we have $q(y')>0$ for all $y'<y$ including $y'=x$. Thus, we have $p(x)-q(x)=0-q(x) < 0$. Now because both $x$ and $y$ have negative sign and are the smallest zeros, we know that $\lim_{z'\to -\infty} p(z')= \lim_{z'\to -\infty} q(z') = \infty$. But because $p$ has higher degree than $q$, there must exist some $z'<x$ such that $p(x)-q(x) > 0$. Hence, $p-q$ has a sign flip and by the intermediate value theorem, there must exist some $z<x$ with $p(z)-q(z) = 0$, as desired.
    The proof of (iii) is analogous to (ii).
\end{proof}

\begin{remark}\label{rmk:interlacingrobust}
    Because \cref{lem:interlacingfacts} only relies on the sign and locations of the zeros of $p$ and $q$ as well as their asymptotic growth, any zero guaranteed by \cref{lem:interlacingfacts} continues to be guaranteed after multiplying either of them by some positive function. We can state this as follows. Let $p,q$ be as above and $f,g >0$ be positive continuous functions with $g(x) = O(f(x))$ as $\abs x\to \infty$. Then, \cref{lem:interlacingfacts} guarantees the same number of zeros for $p-q$ and $fp-gq$, with the same bounds.
\end{remark}

We can now use this machinery to find the zeros of $P_n$ as well as those of $P_n + P_{n-1}$.
\begin{lemma}\label{lem:pzeros}
    For any $\gamma>0$, $P_n$ has $n$ real zeros in $(-1,1)$.
\end{lemma}
\begin{proof}
    We first recall that $P_n = U_n + e^{-\gamma/2}U_{n-1}$. We will now use the previous lemma to look for zeros of $U_n - (-e^{-\gamma/2}U_{n-1})$. We will denote respectively the roots of $U_n$ and $-e^{-\gamma/2}U_{n-1}$ as
    \begin{gather*}
        x_k = \cos\left(\frac{n+1-k}{n+1}\pi\right), \ k=1,\dots,n, \quad y_k = \cos\left(\frac{n-k}{n}\pi\right),\ k=1,\dots,n-1,
    \end{gather*}and have $-1<x_1<y_1<x_2<y_2<\dots<x_{n-1}<y_{n-1}<x_n<1$. 
    We assume without loss of generality that $n$ is even and we have
    \begin{gather*}
        \sgn x_k = (-1)^k,\quad \sgn y_k = (-1)^{k},
    \end{gather*} because $U_n$ has positive leading coefficient and $-e^{-\gamma/2}U_{n-1}$ has negative leading coefficient. 
    Hence, the zeros of $U_n$ and $-e^{-\gamma/2}U_{n-1}$ are fully interlaced and we can use \cref{lem:interlacingfacts}(i) $2n-1$ times to get $2n-1$ zeros, $(z_i)_{i=2}^{2n}$, of $U_n - (-e^{-\gamma/2}U_{n-1})$:
    \[
        -1<x_1<y_1<z_2<x_2<y_2<z_3<\dots<z_{n-1}<x_{n-1}<y_{n-1}<z_n<x_n<1,
    \]
    which are already bounded in $(-1,1)$.
    Furthermore, we can use \cref{lem:interlacingfacts}(ii) to get another zero of $U_n - (-e^{-\gamma/2}U_{n-1})$, $z_1 < x_1< y_1$.

    It remains to show that also $-1<z_1$. We use
    \[
        U_k(-1) = (-1)^k(k+1),
    \] and the fact that $n$ is even to get \[
        U_n(-1) = n+1 > e^{-\gamma/2}U_{n-1}(-1) = e^{-\gamma/2}n.
    \] By the same argument as the one in \cref{lem:interlacingfacts}, the sign of $U_n - (-e^{-\gamma/2}U_{n-1})$ must have flipped between $-1$ and $x_1$ and thus, $z_1\in (-1,x_1)$. This concludes the proof.
\end{proof}

\begin{figure}
    \centering
    \includegraphics[width=0.5\textwidth]{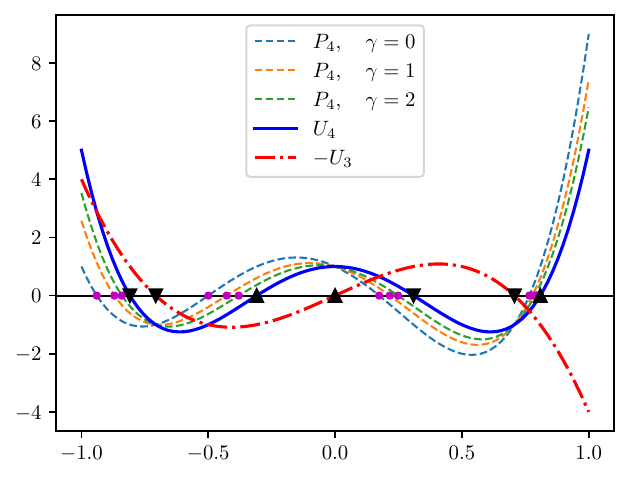}
    \caption{Illustration of the result in \cref{lem:pzeros}. The Chebyshev polynomials of the second order, $U_4$ and $U_3$, are shown in blue and red with their respective zeros marked by triangles. The orientation of these triangles marks the sign of the corresponding zero. The polynomials $P_4$ for different values of $\gamma$ are drawn in dashed lines for various values of $\gamma$, with their zeros marked in purple dots. For $\gamma=0$, the zeros of $P_4$ are exactly the intersections of $U_4$ and $-U_3$.
    We can see that, independently of $\gamma$, the zeros of $P_4$ always occur between two zeros of $U_4$ and $U_3$ of opposite signs, or to the left of the smallest zero of $U_4$ --- as predicted by \cref{lem:interlacingfacts}.}
    \label{fig:chebyshev_interlacing}
\end{figure}

\begin{lemma}\label{lem:psumzeros}
    $P_n+P_{n-1}$ has at least $n-1$ real zeros in $(-1,1)$ for any $\gamma>0$.
\end{lemma}
\begin{proof}
    We use the definition of $P_n$ and the recursion formula for the Chebyshev polynomials to get
    \begin{gather*}
        P_n+P_{n-1} = 0 \iff U_n + e^{-\gamma/2}U_{n-1} + U_{n-1} + e^{-\gamma/2}U_{n-2} = 0\\
        \iff (2x+e^{-\gamma/2}+1)U_{n-1}-\underbrace{(1-e^{-\gamma/2})}_{> 0}U_{n-2}=0.
    \end{gather*}
    It is thus equivalent to look for intersections of $S_1\coloneqq (2x+e^{-\gamma/2}+1)U_{n-1}$ and $S_2 \coloneqq (1-e^{-\gamma/2})U_{n-2}$. We know that the zeros of $S_1$ are $x_k = \cos\left(\frac{n-k}{n}\pi\right), \ k=1,\dots,n-1$ and $x^* = -\frac{1+e^{-\gamma/2}}{2}$,  while the zeros of $S_2$ are $y_k = \cos\left(\frac{n-1-k}{n-1}\pi\right),\ k=1,\dots,n-2$.

    We note that the zeros $x_k$ and $y_k$ do not depend on $\gamma$. Hence, only $x^*$ depends on $\gamma$ and moves from $-1$ to $-\frac{1}{2}$ as $\gamma\to \infty$. For $\gamma>0$ small enough, we have
    \[
        x^*<x_1<y_1<\dots<x_{n-2}<y_{n-2}<x_{n-1}.
    \]
    We assume without loss of generality that $n$ is even and use the fact that $S_1$ and $S_2$ have positive leading coefficients to get \[
        \sgn x_k = (-1)^{k+1},\quad \sgn y_k = (-1)^{k}.
    \]
    This allows us to use \Cref{lem:interlacingfacts}(i) in order to find $n-2$ zeros of $S_1-S_2$, $z_1,\dots , z_{n-2} \in (-1,1)$. Furthermore, we use $\deg S_1> \deg S_2$ together with the fact that $y_{n-2}<x_{n-1}$ both have positive sign to get another zero $y_{n-2}<x_{n-1}<z_{n-1}$ of $S_1-S_2$ by  \Cref{lem:interlacingfacts}(iii). The results thus hold for $\gamma>0$ small enough.

    It remains to show that increasing $\gamma>0$ leaves the number of such zeros unchanged. If we gradually increase $\gamma$ from zero to infinity, then $x^*$ moves from $-1$ to $-\frac{1}{2}$ and only one of the three following statements hold for a small enough change in $\gamma$:
    \begin{enumerate}
        \item[(i)] $x^*$ does not cross any zero $x_k$ or $y_k$;
        \item[(ii)] $x^*$ passes through a zero $x_k$ of $S_1$;
        \item[(iii)] $x^*$ passes through a zero $y_k$ of $S_2$.
    \end{enumerate}
    In the first case, no zero changes sign and the order of zeros is unaffected. Because of this, the conditions for \Cref{lem:interlacingfacts} remain exactly the same and we continue to find $n-1$ zeros of $S_1-S_2$, although possibly at slightly different locations.

    For the second case, we move from $x^*<x_k$ to $x_k<x^*$. We assume without loss of generality that $\sgn x^* = -1, \sgn x_k = 1$ for $x^*<x_k$. Because the signs of zeros alternate, $x_k$ and $x^*$ change sign after this interaction. But the big picture remains unchanged as $S_1$ continues to have a zero of negative sign followed by a zero of positive sign and thus leaving the number of zeros the same.

    Finally, in the third case, we move from $x^*<y_k<x_{k+1}$ to $y_k<x^*<x_{k+1}$. We assume without loss of generality that $\sgn x^* = -1$ which by the above argument makes $\sgn x_{k+1}=1$ and $\sgn y_k=1$. Lemma \ref{lem:interlacingfacts}(i) then delivers a zero $z$ of $S_1-S_2$ with $x^*<z<y_k$. As $x^*$ passes through $y_k$ no sign change occurs, since $x^*$ and $y_k$ belong to different polynomials. We can thus continue to apply \Cref{lem:interlacingfacts}(i) to get a zero $z$ of $S_1-S_2$ with $y_k<z<x^*$ and the total amount of zeros of $S_1-S_2$ remains unchanged.

    This concludes the proof. We refer to \cref{fig:interlacing_transition} for an illustration of second and third kind transitions.
\end{proof}

\begin{figure}[!h]
    \centering
    \includegraphics[width=\textwidth]{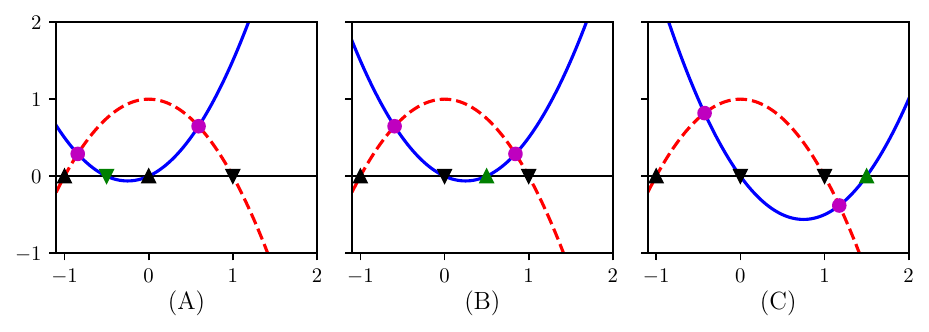}
    \caption{Illustration of the main proof idea in \cref{lem:psumzeros}. The two polynomials in solid blue and dashed red symbolise $(2x+e^{-\gamma/2}+1)U_{n-1}$ and $(1-e^{-\gamma/2})U_{n-2}$, respectively. Their zeros are marked by triangles with orientation determined by their signs. Intersections of these polynomials then correspond to zeros of $P_n+P_{n-1}$ and are marked as purple circles.
    As we move from (A) to (B) to (C), $\gamma$ is increased and the special zero $x^*$, marked in green moves to the right while the other zeros remain stationary. From  (A) to (B),  a transition of the second kind occurs, and from (B) to (C), a transition of the third kind occurs, as described in the proof of \cref{lem:psumzeros}. Notably, both transformations leave the total number of zeros unchanged.}
    
    \label{fig:interlacing_transition}
\end{figure}
We can now combine the results of this subsection to finally prove the desired statements.
\begin{proof}[Proof of \cref{prop:Szeros}]
    Recall that  \[
    \mc S(\lambda) = P_N(\mu(\lambda))P_N(\mu(-\lambda))-e^\gamma P_{N-1}(\mu(\lambda))P_{N-1}(\mu(-\lambda))=0.
\]
As mentioned above, $\mc S(\lambda)$ is even and its zeros must be symmetric about the origin. This allows us to focus on $\lambda\geq 0$ without loss of generality.
Furthermore, because $\ker \cg \neq \emptyset$, we know that $\lambda = 0$ must be a zero of $\mc S$. Moreover, because $\mc{S}(\lambda)$ is an even polynomial, $\lambda=0$ must be a double zero.

It thus remains to show that $\mc S(\lambda)$ has $N-1$ real, distinct and positive zeros. We will assume that $\lambda>0$ for the rest of this proof. 

We now aim to prove that $P_N(\mu(-\lambda))$ and $-e^\gamma P_{N-1}(\mu(-\lambda))$ are always positive for $\lambda>0$. Later, \cref{rmk:interlacingrobust} will allow us to ignore these factors.
Recall that
\begin{equation}\label{equ:proofrootsofs1}
    \mu(\lambda) = \frac{\lambda-\alpha}{2\sqrt{\beta\eta}} =\lambda \frac{1}{\gamma}\sinh \frac{\gamma}{2} - \cosh \frac{\gamma}{2}.
\end{equation}
Because $\cosh \frac{\gamma}{2} > 1$ for all $\gamma>0$, we have $\mu(-\lambda) < -\cosh \frac{\gamma}{2} < -1$. But, from the previous corollary, we know that all zeros of $P_N$ and $P_{N-1}$ lie in $(-1,1)$. If we assume that $N$ is even, then, without loss of generality, we can  conclude that
\[
    P_{N-1}(\mu(-\lambda)) < 0 < P_N(\mu(-\lambda)) \quad \text{for }\lambda>0.
\]
This ensures that $P_N(\mu(-\lambda))$ and $-e^\gamma P_{N-1}(\mu(-\lambda))$ are positive for $\lambda>0$. 

On the other hand, by the same argument we can see that \cref{lem:psumzeros} guarantees $N-1$ positive zeros $\lambda>0$ for $P_N(\mu(\lambda))+P_{N-1}(\mu(\lambda))$. More concretely, by \cref{lem:psumzeros} there are $N-1$ roots $\mu\in (-1,1)$ of $P_N(\mu)+P_{N-1}(\mu)$. By (\ref{equ:proofrootsofs1}), this corresponds to $N-1$ positive $\lambda$. 

Finally, because the $N-1$ distinct zeros as well as the bounds that \cref{lem:psumzeros} guarantees for $P_N(\mu(\lambda))+P_{N-1}(\mu(\lambda))$ stem from \cref{lem:interlacingfacts}, \cref{rmk:interlacingrobust} states that they are also guaranteed for $\mc{S}(\lambda)$, and we find  the desired $N-1$ distinct positive zeros of $\mc{S}(\lambda)$.
\end{proof}

\section{Technical proofs}\label{sec:technical_proofs}
\begin{proof}[Proof of \cref{lem:acharact}]
    The fact that $a$ is holomorphic away from $[-1,1]$ follows from the fact that for $\mu$ outside $[-1,1]$ the two square roots incur their branch cuts simultaneously which cancels them out.

    We now investigate the inverse and set $a(\mu)$ equal to some $z = re^{\i\varphi}$.
    We have
    \begin{gather*}
        a(\mu) = \mu + \sqrt{\mu+1}\sqrt{\mu-1} = re^{\i\varphi} \implies \mu^2-1 = r^2e^{i2\varphi} - 2re^{\i\varphi}\mu + \mu ^2 \\\iff \mu = \frac{1}{2}(re^{\i\varphi}+\frac{1}{r}e^{-\i\varphi}).
    \end{gather*} The first implication occurs because we move over $\mu$ and then square both sides of the equation. We have thus identified a potential inverse in $\mu(re^{\i\varphi}) \coloneqq \frac{1}{2}(re^{\i\varphi}+\frac{1}{r}e^{-\i\varphi})$. Therefore, 
    \[
        \Re \mu(re^{\i\varphi}) = \frac{r^2+1}{2r}\cos \varphi,\quad \Im \mu(re^{\i\varphi}) =\frac{r^2-1}{2r}\sin \varphi.
    \]
    Now, we plug this potential inverse $\mu(re^{\i\varphi})$ into $a$ to check if it is actually one. We assume that $0 < \varphi < \frac{\pi}{2}$ and treat the cases $r>1$, $r<1$ and $r=1$, separately. In the first case, $\mu(re^{\i\varphi})$ lies in the first quadrant which ensures that $$a(\mu(re^{\i\varphi})) = \mu(re^{\i\varphi}) + \sqrt{\mu(re^{\i\varphi})+1}\sqrt{\mu(re^{\i\varphi})-1} = \mu(re^{\i\varphi}) + \sqrt{\mu(re^{\i\varphi})^2-1}.$$ We then plug in the definition of $\mu(re^{\i\varphi})$ to get $$a(\mu(re^{\i\varphi})) = \frac{1}{2}(re^{\i\varphi}+\frac{1}{r}e^{-\i\varphi}) + \sqrt{(\frac{1}{2}(re^{\i\varphi}-\frac{1}{r}e^{-\i\varphi}))^2}.$$ Analogous arguments as above show that $\frac{1}{2}(re^{\i\varphi}-\frac{1}{r}e^{-\i\varphi})$ is again in the first quadrant as long as $r>1$. This allows us to cancel the root with the square and to get $a(\mu(re^{\i\varphi})) = \frac{1}{2}(re^{\i\varphi}+\frac{1}{r}e^{-\i\varphi}) + \frac{1}{2}(re^{\i\varphi}-\frac{1}{r}e^{-\i\varphi}) = re^{\i\varphi}$, as desired.

    We now move to the second case where $r<1$. This gives a negative sign to $\frac{r^2-1}{2r}<0$ and shows that $\mu(re^{\i\varphi})$ lies in the fourth quadrant as a consequence. The first consolidation of roots then works the same as above. However, once we get to $\frac{1}{2}(re^{\i\varphi}-\frac{1}{r}e^{-\i\varphi})$ we see that it now lies in the second quadrant. Thus, we incur a negative sign when eliminating the root and get to $a(\mu(re^{\i\varphi})) = \frac{1}{2}(re^{\i\varphi}+\frac{1}{r}e^{-\i\varphi}) - \frac{1}{2}(re^{\i\varphi}+\frac{1}{r}e^{-\i\varphi}) = \frac{1}{r}e^{-\i\varphi} \neq re^{\i\varphi}$. Hence, in this case, $\mu(re^{\i\varphi})$ is not an inverse. Because this was the only candidate, we can conclude that there exists no $\mu\in \C$ such that $a(\mu)$ has absolute value less than one.

    In the case where $r=1$, we have $\Re \mu(re^{\i\varphi}) = \cos \varphi, \Im \mu(re^{\i\varphi}) =0$. After plugging this into $a$, analogous arguments as above show that $a(\mu) = e^{\i\varphi}$ for some $\mu \in \C$ if and only if $\mu = \cos \varphi$ and $0\leq \varphi \leq \pi$.

    For $a$ defined as above, the case $r>1$ shows that it is injective, $r<1$ shows it is surjective and $r=1$ characterises the degenerate region.

    The fact that the inverse is holomorphic can be seen immediately from its form $z\mapsto \frac{1}{2}(z+z^{-1})$ because we are away from zero.
\end{proof}

\begin{proof}[Proof of \cref{lem:aconv}]
    Recall that we can write the Chebyshev polynomials as follows
    \[
        U_n(\mu) = \frac{a(\mu)^{n+1}-a(\mu)^{-(n+1)}}{2\sqrt{\mu+1}\sqrt{\mu-1}}.
    \]
    Using this fact and $P_n = U_n + e^{\frac{-\gamma}{2}}U_{n-1}$, we find that
    \begin{gather*}
        \frac{P_n(\mu)}{P_{n-1}(\mu)} = \frac{a(\mu)^{n+1}-a(\mu)^{-(n+1)}+e^{\frac{-\gamma}{2}}a(\mu)^{n}-e^{\frac{-\gamma}{2}}a(\mu)^{-n}}{a(\mu)^{n}-a(\mu)^{-n}+e^{\frac{-\gamma}{2}}a(\mu)^{n-1}-e^{\frac{-\gamma}{2}}a(\mu)^{-(n-1)}},
    \end{gather*}
    which after some algebraic manipulation yields
    \begin{align*}
        \abs{\frac{P_n(\mu)}{P_{n-1}(\mu)}-a(\mu)}= \abs{a(\mu)}^{-2n+2}\frac{
        \abs{1-a(\mu)^{-2}}
        \abs{a(\mu)^{-1}+e^{-\frac{\gamma}{2}}}
        }{
        \abs{1+a(\mu)^{-1}e^{-\frac{\gamma}{2}}-a(\mu)^{-2n+1}(a(\mu)^{-1}+e^{-\frac{\gamma}{2}})}
        }.
    \end{align*}
    By \cref{lem:acharact}, we know that $\abs{ a(\mu)} \geq 1$ which we can use in the above inequality to obtain that
    \begin{gather*}
        \left|\frac{P_n(\mu)}{P_{n-1}(\mu)}-a(\mu)\right|\leq \abs{a(\mu)}^{-2n+2}\left(2\frac{
            1+e^{-\frac{\gamma}{2}}
        }
        {
            1-\abs{a(\mu)}^{-1}e^{-\frac{\gamma}{2}}-2\abs{a(\mu)}^{-2n+1}
        }\right).
    \end{gather*}
    The condition $\abs{a(\mu)}^{-2n+2}<\frac{e^\gamma}{2}$ ensures that the denominator in the above fraction is always larger than zero.

    We can now use this inequality to prove uniform convergence.
    By \cref{lem:acharact},  for any $\varepsilon>0$, $U_\varepsilon\coloneqq \{\mu\in \C\mid \abs{a(\mu)}< 1+\varepsilon\}$ is an arbitrarily small neighbourhood of $[-1,1]$. We now fix $\varepsilon>0$ arbitrarily small and look at the complement $D_\varepsilon = \C \setminus U_\varepsilon$. By definition,  we know that $\abs{a(\mu)} \geq 1+\varepsilon$ on $D_\varepsilon$. Therefore,  we get
    \begin{gather*}
        \left|\frac{P_n(\mu)}{P_{n-1}(\mu)}-a(\mu)\right|\leq (1+\varepsilon)^{-2N+2}\left(2\frac{
        1+e^{-\frac{\gamma}{2}}
        }
        {
        1-(1+\varepsilon)^{-1}e^{-\frac{\gamma}{2}}-2(1+\varepsilon)^{-2n+1}
        }\right)
    \end{gather*} if we choose $n\in \N$ large enough such that $(1+\varepsilon)^{-2n+2}<\frac{e^\gamma}{2}$.
    This bound is independent of $\mu$ and goes to zero as $n\to \infty$. Hence, the convergence must be uniform.
\end{proof}
\printbibliography

\end{document}